\documentclass[11pt]{article}

\usepackage[margin=1in]{geometry}
\usepackage{setspace}
\usepackage{fancyhdr}

\usepackage[utf8]{inputenc}
\usepackage{amsmath,amssymb,amsfonts,amsthm}
\usepackage{bbm}
\usepackage{graphicx}
\usepackage{fullpage}
\usepackage[backref=page]{hyperref}
\usepackage{color}
\usepackage{wrapfig}
\usepackage{tikz}
\usepackage{setspace}
\usepackage[noend]{algpseudocode}
\usepackage[framemethod=tikz]{mdframed}
\usepackage{xspace}
\usepackage{pgfplots}
\usepackage{framed}
\usepackage{subcaption}
\usetikzlibrary{arrows.meta,positioning, calc}
\usepackage{cleveref}
\usepackage[ruled,vlined]{algorithm2e}
 \pgfplotsset{compat=1.5}

\def\final{0}  % set this to 1 to get a comment-free version
\def\iflong{\iffalse}
\ifnum\final=0  %namely if we allow comments in the output
\newcommand{\tanote}[1]{{\color{purple}[\small {Tamalika: \bf #1}]}}%\marginpar{\color{red}*}}}
\else % in this case [final=1] we don't want any comments to show
\newcommand{\tanote}[1]{}
\fi

\newcommand{\Lap}{\ensuremath{{\textsf{Lap}}}\xspace}
\newcommand{\E}[0]{\mathop{\bbE}\xspace}
\newcommand{\bfc}{\ensuremath{\mathbf{c}}\xspace}

\DeclareMathOperator*{\poly}{poly}

  \newcommand{\eps}[0]{\ensuremath{\varepsilon}}
  \let\epsilon\eps
  
  \newcommand{\cA}{\ensuremath{{\mathcal A}}\xspace}
  \newcommand{\cB}{\ensuremath{{\mathcal B}}\xspace}
  \newcommand{\cC}{\ensuremath{{\mathcal C}}\xspace}
  
  \newcommand{\cE}{\ensuremath{{\mathcal E}}\xspace}
  \newcommand{\cF}{\ensuremath{{\mathcal F}}\xspace}

  \newcommand{\cO}{\ensuremath{{\mathcal O}}\xspace}

  \newcommand{\cS}{\ensuremath{{\mathcal S}}\xspace}
  
  \newcommand{\cU}{\ensuremath{{\mathcal U}}\xspace}

  \newcommand{\cX}{\ensuremath{{\mathcal X}}\xspace}
  \newcommand{\cY}{\ensuremath{{\mathcal Y}}\xspace}

  \newcommand{\bbE}{\ensuremath{{\mathbb E}}\xspace}

  \newcommand{\bbR}{\ensuremath{{\mathbb R}}\xspace}

  \newcommand{\bbZ}{\ensuremath{{\mathbb Z}}\xspace}

\makeatletter
\newtheorem*{rep@theorem}{\rep@title}
\newcommand{\newreptheorem}[2]{%
\newenvironment{rep#1}[1]{%
 \def\rep@title{\theoremref{##1} Restated}%
 \begin{rep@theorem}}%
 {\end{rep@theorem}}}
\makeatother

\makeatletter
\newtheorem*{rep@lemma}{\rep@title}
\newcommand{\newreplemma}[2] {%
\newenvironment{rep#1}[1]{%
 \def\rep@title{\lemmaref{##1} Restated}%
 \begin{rep@lemma}}%
 {\end{rep@lemma}}}
\makeatother

\newtheorem{theorem}{Theorem}
\newreptheorem{theorem}{Theorem}

\newtheorem{definition}{Definition}
\newtheorem{lemma}{Lemma}
\newtheorem{proposition}{Proposition}
\newreplemma{lemma}{Lemma}

\theoremstyle{definition}
\newtheorem*{remark}{Remark}
\crefname{theorem}{Theorem}{Theorems}
\crefname{lemma}{Lemma}{Lemmas}
\crefname{definition}{Definition}{Definitions}
\crefname{proposition}{Proposition}{Propositions}
\crefname{corollary}{Corollary}{Corollaries}

\crefname{algocf}{Algorithm}{Algorithms}
\Crefname{algocf}{Algorithm}{Algorithms}
\crefname{section}{Section}{Sections}
\crefname{figure}{Figure}{Figures}
\crefname{table}{Table}{Tables}
\crefname{appendix}{Appendix}{Appendixes}

\title{Differentially Private Clustering in Data Streams}
\date{}

\author{
Alessandro Epasto%\thanks{E-mail: \email{aepasto@google.com}}
\\Google 
\and
Tamalika Mukherjee\thanks{Tamalika Mukherjee was supported in part by the Bilsland Dissertation Fellowship. Work partially done while a Ph.D. Student at Purdue University and a Student Researcher at Google.}
\\Columbia University
\and
Peilin Zhong%\thanks{E-mail: \email{peilinz@google.com} }
\\Google 
}

\begin{document}

\maketitle
%\tableofcontents 
%\newpage 

\begin{abstract}

Clustering problems (such as $k$-means and $k$-median) are fundamental unsupervised machine learning primitives, and
streaming clustering algorithms have been extensively studied in the past.
However, since data privacy becomes a central concern in many real-world applications, non-private clustering algorithms may not be as applicable in many scenarios.

In this work, we provide the first differentially private algorithms for $k$-means and $k$-median clustering of $d$-dimensional Euclidean data points over a stream with length at most $T$ using space that is \emph{sublinear} (in $T$) in the \emph{continual release} setting where the algorithm is required to output a clustering at every timestep.

We achieve (1) an $O(1)$-multiplicative approximation with $\tilde{O}(k^{1.5} \cdot \poly(d,\log(T)))$ space and $\poly(k,d,\log(T))$ additive error, or (2) a $(1+\gamma)$-multiplicative approximation with $\tilde{O}_\gamma(\poly(k,2^{O_\gamma(d)},\log(T)))$ space for any $\gamma>0$, and the additive error is \sloppy $\poly(k,2^{O_\gamma(d)},\log(T))$.
Our main technical contribution is a differentially private clustering framework for data streams which only requires an offline DP coreset or clustering algorithm as a blackbox.
%In addition, our algorithmic framework is also differentially private under the continual release setting, i.e., the union of outputs of our algorithms at every timestamp is always differentially private.

\end{abstract}
\thispagestyle{empty}
\newpage
\setcounter{page}{1}
\pagestyle{plain}
\section{Introduction}

In real-world applications, a major challenge in dealing with large-scale data is that the entire datasets are too large to be stored in the computing system.
The need to address this challenge and the success of large-scale systems (such as Spark Streaming~\cite{zaharia2012discretized}) that process data in streams have driven the study of the \emph{streaming model}, introduced by the seminal work~\cite{alon1999space}.
In this model, there is a stream of data points --- a data point arrives at each timestamp, and a streaming algorithm can only access these data points by a single pass.
The goal is to output an (approximate) solution to a problem with respect to the set of data points that are arrived while using as small space as possible.
If an algorithm is required to output at every timestamp when a new data point arrives, then it is called the \emph{continual release} setting.
Otherwise, the algorithm only needs to output at the end of the stream, which is called the \emph{one-shot} setting.
The streaming model has attracted a lot of attention from different areas in the past decades.
In particular, streaming clustering algorithms have been extensively studied by the clustering literature.

Clustering is an essential primitive in unsupervised machine learning, and its geometric formulations, such as $k$-means and $k$-median, have been studied extensively, e.g.,~\cite{AGKMMP01, charikar2002constant, har2004coresets, chen2006k, chen2008constant, awasthi2010stability, ostrovsky2013effectiveness, li2016approximating, ahmadian2019better}. 

In the streaming $k$-clustering problem, there is a stream of points in $\mathbb{R}^d$, and the goal is to output a set of $k$ centers at each timestamp $t$, and minimize the $k$-clustering cost with respect to data points arrived before $t$. There is a long list of work (e.g., a subset includes~\cite{har2004coresets,har2005smaller,Chen09,feldman2011unified,braverman2019streaming,cohen2022towards,cohen2023streaming}) studying $k$-means and $k$-median problem in the streaming setting. 
The state-of-the-art result is achieved by~\cite{cohen2023streaming} which uses $\tilde{O}\left(\frac{kd}{\gamma^2}\right)\cdot \min\left(\frac{1}{\gamma^z},k\right)\cdot \poly(\log\log T)$ space to obtain a $(1+\gamma)$-approximation with probability $0.9$ at the end of the stream.
However, none of these algorithms are private, which means they are not applicable when the dataset involves personal information and privacy is considered to be a major concern in real-world applications.

Differential privacy (DP)~\cite{DMNS06} has become the de facto standard for preserving data privacy due to its compelling privacy guarantees and mathematically rigorous definition. 
DP $k$-clustering algorithms in the offline setting have been studied for years~\cite{NissimRS07,FeldmanFKN09,FeldmanXZR17, GuptaLMRT10,BalcanDLMZ17, HuangL18,NissimS18, StemmerK18,GhaziKM20,cohen21c}, where the main focus is to improve the approximation ratio and achieve fast sequential running time.
DP $k$-clustering problem has also been studied in other computational models which are more relevant to large-scale computations such as sublinear-time~\cite{blocki2021differentially} and massively parallel computing (MPC)~\cite{CELMMSSV22, EZNMC22}, and distributed computing setting~\cite{xia2020distributed,chang2021locally}.
However, the landscape of the DP $k$-clustering problem in the streaming model is still mysterious. 
In fact, to the best of our knowledge, there is no previously known DP streaming algorithm achieving $O(1)$-multiplicative error using sublinear space even in the one-shot setting.

In this work, we present the \emph{first} DP streaming algorithms for Euclidean $k$-means and $k$-median clustering using $\poly(k,d,\log(T))$ space to achieve an $O(1)$-multiplicative error and a $\poly(k,d,\log(T))$-additive error. 
Our algorithms are DP under the \emph{continual release} setting.
Note that any DP algorithm under the continual release setting is always DP under the one-shot setting.

\subsection{Differential Privacy  Model and Clustering Problem}
We first formally define differential privacy and the streaming model under continual release setting and then define the clustering problem. The input is a stream of points $x_1,x_2,\cdots,x_T\in\mathbb{R}^d$, where each $x_i\in\mathbb{R}^d$ satisfies $\|x_i\|_2\leq \Lambda$, i.e., we assume all input points are within a ball of radius $\Lambda$. We work in the \emph{insertion-only} setting, where points arrive sequentially but never leave the stream. That is, at time $t$ the dataset consists of all points ${x_1,\ldots,x_t}$ seen so far, and future outputs must remain consistent with this growing prefix.

In this paper, we study streaming algorithms $\cA$ under the continual release setting, i.e., the entire output of $\cA$ is $(s_1,s_2,\cdots,s_T)$ where $s_t$ is the output of $\cA$ at timestamp $t$ with respect to the data arrived no later than $t$.
We work under \emph{event-level privacy} where the streams $\cS=(x_1,\ldots, x_T)$ and $\cS'=(x'_1,\ldots, x'_T)$ are {neighboring} if there exists at most one timestamp $t^* \in [T]$ for which $x_{t^*} \neq x'_{t^*}$ and $x_{t} = x'_{t}$ for all $t \neq t^*$.  

If we do not specify the timestamp, then the output of $\cA(\cS)$ indicates the entire output of $\cA$ over the stream $\cS$ at every timestamp.

\begin{definition}[Differential privacy \cite{DMNS06}]\label{def:DP}
A randomized algorithm $\cA$ is $(\eps,\delta)$-DP if for every pair of neighboring streams $\cS \sim \cS'$, and for all sets $\cO$ of possible outputs, we have that $\Pr[\cA(\cS) \in \cO ] \leq e^\eps \Pr [\cA(\cS') \in \cO] +\delta  $. When $\delta=0$ we simply say that the algorithm is $\eps$-DP.
\end{definition}

 For points $x,y \in \bbR^d$, we let $d(x,y) = \|x-y\|_2$ be the Euclidean distance between $x$ and $y$. Given a set $\cC$, we define $d(x,C):= \min_{c \in C} d(x,c)$. 

For a set of centers $C$, we define the cost of clustering for the set $\cS$ w.r.t. $C$ as  
$$\textsf{cost}(C, \cS) = \sum_{x \in \cS} d^z(x,C)$$
where $z=1$ for $k$-median, and $z=2$ for $k$-means.

Our goal in DP clustering is to produce a set of $k$ centers $C_\cS$ for input stream $\cS$ such that (1) $C_\cS$ is $(\eps,\delta)$-DP wrt $\cS$, and (2) $\textsf{cost}(C_\cS, \cS) \leq \alpha \cdot \textsf{cost}(C^{opt}_\cS, \cS) + \beta$.

\subsection{Our Results}
Before presenting our main results, we first define some useful terminology. The notation $O_x(\cdot)$ ignores factors involving $x$. We use the term \emph{semicoreset} to describe a relaxed version of a coreset that allows an additive error proportional to the optimal cost. This notion is not intended as a canonical definition, but rather as a convenient abstraction that simplifies the presentation of our results. In particular, it enables a unified statement of \cref{thm:main-inf} that can work using both static DP coreset algorithms and static DP clustering algorithms as a black-box.

\begin{definition}[$(\kappa,\eta_1,\eta_2)$-semicoreset]\label{def:semicoreset}
Given a point set $P$ in $\bbR^d$, a (multi)set $Q$ is a $(\kappa,\eta_1,\eta_2)$-semicoreset of $P$ for $k$-clustering ($k$-means or $k$-median) if 
$$ \frac{1}{\kappa}\cdot \textsf{cost}(C,P)-\eta_1 \cdot \textsf{cost}(C^{opt}_P,P) -\eta_2 \leq \textsf{cost}(C,Q) \leq \kappa \cdot \textsf{cost}(C,P)+\eta_2 $$
for any set of $k$-centers $C \subseteq \bbR^d$.
\end{definition}

In particular, a coreset that has additive error that is not proportional to the optimal cost is simply a $(\kappa,0,\eta_2)$-semicoreset. We also abuse notation and refer to a coreset with \emph{no} additive error, i.e., $(\kappa,0,0)$-semicoreset as a \emph{$\kappa$-coreset}. 

Given an offline DP $k$-clustering algorithm one can obtain an offline DP semicoreset by using a transformation (stated in \cref{thm:clustering-to-semicoreset}) by~\cite{EZNMC22} --- we use this transformation as a blackbox in the sequel. 
We state our results for $k$-means, but note that our results easily generalize to $k$-median. As is standard in DP clustering literature~(e.g.,\cite{StemmerK18}), we assume $\Lambda$ is an upper bound on the diameter of the space of input points. For ease of presentation, we assume $\Lambda=1$ in this part.

Our main result is a general framework for DP $k$-clustering in the streaming setting which utilizes an offline DP semicoreset algorithm as a black-box (see \cref{thm:main-inf}). Using existing results from the DP clustering literature, the cost of the resulting clustering output by our framework achieves (1) an $O(1)$-multiplicative error with space complexity having a $k^{1.5}$ dependency  --- using the DP clustering algorithm from~\cite{StemmerK18} (see~\cref{thm:final-stemmer}), or (2) a $(1+\gamma)$-multiplicative error with space complexity having a $\poly(k)$ dependency --- using the DP coreset algorithm from~\cite{GhaziKM20}  (see~\cref{thm:final-ghazi}). 

We emphasize that our work establishes the theoretical feasibility of sublinear-space DP clustering in streams—a fundamental question that was previously open. 
As offline DP clustering algorithms achieve better additive error bounds, our black-box approach immediately inherits these improvements, potentially leading to more practical streaming variants.

We assume we are given a non-DP algorithm in the offline setting that can compute a $(1+\gamma)$-approximation to $k$-means (e.g.~\cite{SaulpicCF19}) in the following statements.

\begin{theorem}[Main]\label{thm:main-inf}
 Given dimension $d$, clustering parameter $k$, arbitrary parameter $C_M$, a non-DP $(1+\gamma)$-coreset algorithm, an $(\eps,\delta)$-DP $(\kappa,\eta_1,\eta_2)$-semicoreset algorithm $\cA$ that outputs a semicoreset of size $SZ_\cA(\cdot)$ and using space $S_\cA(\cdot)$. 
 Then, there exists a streaming algorithm $\cA'$ for $k$-means that outputs a set of centers $\cC_{\hat{\cY}}$ at every timestep $t \in [T]$ such that
\begin{enumerate}
    \item \label{it:dp-main}(Privacy) $\cA'$ is $(3\eps,\delta)$-DP under the continual release setting. 
    \item \label{it:acc-main}(Accuracy) With probability $1-\frac{1}{T^2}-\frac{1}{k^{O(\poly(k)}}$, 
  \begin{align*}
  \mathsf{cost}(\cC_{\hat{\cY}},\cS) 
  &\leq {\kappa }\cdot \left((1+\gamma)C_M +\eta_1% \Lambda^2 
  +\frac{(1+\gamma)^4}{(1-\gamma)^3}(\kappa+C_M)\right) \cdot \mathsf{cost}(\cC^{opt}_\cS, \cS)+ V(d,k,\eps,\delta, T,\gamma)
\end{align*}
where $\cS$ denotes the set of all points given by the stream before timestamp $t$, $\cC_{\cS}^{opt}$ is the optimal solution, and $V(d,k,\eps,\delta,T,\gamma) =  \Tilde{O}_\gamma(\kappa (\frac{C_M k^2}{d} +{kM})\frac{1}{\eps}\poly\log(T))$.
\item(Space) \label{it:space-main} $\cA'$ consumes $\Tilde{O}_\gamma(S_\cA(M)+SZ_\cA(M)+M+k\poly\log(\frac{T}{M}))$ space. 
\end{enumerate} 
 where $M=O( \frac{d^3 \eta_2 }{C_M}) $.
\end{theorem}

\begin{remark}
The space complexity achieved by our framework depends on the additive error (i.e., the term $\eta_2$) of the offline DP $(\kappa,\eta_1,\eta_2)$-semicoreset algorithm used as a black-box. In particular, if there exists a DP semicoreset algorithm which achieves a better additive error in the offline setting, then our framework immediately gives a DP clustering algorithm with better space bounds in the continual release setting.
\end{remark}
\begin{remark}
The parameter $C_M$ in \cref{thm:main-inf} is assigned an appropriate value depending on the error of the offline DP semicoreset algorithm used as a black-box in order to obtain the desired multiplicative approximation of the resulting algorithm $\cA'$. For e.g., we can set $C_M=O(1)$ to obtain an $O(1)$-multiplicative approximation in \cref{thm:final-stemmer}, and we can set $C_M=\gamma/2$ to obtain a $(1+O(\gamma))$-multiplicative approximation in \cref{thm:final-ghazi}. 
\end{remark}

We first present \cref{thm:final-stemmer} which is the result of applying our framework to the DP $k$-means algorithm given by~\cite{StemmerK18} in the offline setting that achieves an $O(1)$ multiplicative error. 
We note that~\cite{StemmerK18} gave a DP coreset algorithm which has an additive error with $\poly(k)$ dependency whereas their DP $k$-means algorithm achieves an additive error of $k^{1.5}$. This is why we choose to use their DP $k$-means algorithm and compute a DP semicoreset albeit with a slight loss in privacy to save on space (in terms of $k$) instead of directly using their DP coreset algorithm.

\begin{theorem}\label{thm:final-stemmer}
Given dimension $d$, clustering parameter $k$, privacy parameters $\eps,\delta$, approximation parameter $\gamma>0$. There exists a streaming algorithm $\cA'$ for $k$-means that outputs a set of $k$ centers $\cC_{\hat{\cY}}$ at every timestep $t \in [T]$ such that
\begin{enumerate}
    \item(Privacy) $\cA'$ is $(5\eps,\delta)$-DP under the continual release setting. 
\item(Accuracy) With probability $1-\frac{1}{T^2}-\frac{1}{k^{O(\poly(k))}}$, 
  \begin{align*}
  \mathsf{cost}(\cC_{\hat{\cY}},\cS) 
  &\leq O_{\gamma}(1) \cdot \mathsf{cost}(\cC^{opt}_\cS, \cS)+ V(d,k,\eps,\delta, T)
\end{align*}
where $\cS$ denotes the set of all points given by the stream before timestamp $t$, $\cC^{opt}_\cS$ is the optimal solution, and \sloppy $V(d,k,\eps,\delta,T)= \tilde{O}_\gamma((k^{2.5}d^{3.51})\cdot \poly(\log(T),\log(\frac{1}{\delta}),\frac{1}{\eps}))$.
\item(Space) $\cA'$ consumes $\tilde{O}((k^{1.5} \cdot d^{4.51})\cdot \poly(\log(T),\log(\frac{1}{\delta}),\frac{1}{\eps}))$ 
space. 
\end{enumerate} 
\end{theorem}

Next, we present the result of applying our framework to the DP coreset for $k$-means of~\cite{GhaziKM20} which achieves a multiplicative error of $1+\gamma$ and an additive error of $O_\gamma(\frac{k^2 2^{{O_\gamma}(d)}}{\eps}\poly\log(T))$ in \cref{thm:final-ghazi}. 
\begin{theorem}\label{thm:final-ghazi}
Given dimension $d$, clustering parameter $k$, privacy parameter $\eps$, approximation parameter $\gamma>0$. 
There exists a streaming algorithm $\cA'$ for $k$-means that outputs a set of $k$ centers $\cC_{\hat{\cY}}$ at every timestep $t \in [T]$ such that
\begin{enumerate}
    \item(Privacy) $\cA'$ is $3\eps$-DP under the continual release setting. 
    \item(Accuracy) With probability $1-\frac{1}{T^2}-\frac{1}{k^{O(\poly(k))}}$, 
    \begin{align*}
 \mathsf{cost}(\cC_{\hat{\cY}},\cS) &\leq (1+\gamma)\mathsf{cost}(\cC^{opt}_\cS, \cS)  +  V(d,k,\eps,T)
\end{align*}
where $\cS$ denotes the set of all points given by the stream before timestamp $t$, $\cC^{opt}_\cS$ is the optimal solution, 
$V(d,k,\eps,T) = \tilde{O}_\gamma (\frac{ k^3 d^3 2^{O_\gamma(d)}}{\eps}\cdot\poly\log(T))$.
\item(Space) $\cA'$ consumes $\tilde{O}\left(\poly\left(\frac{2^{O_\gamma(d)}}{\eps} \log(T) ,k,d\right) \right)$ space.
\end{enumerate} 
\end{theorem}

\begin{remark}
We note that the exponential dependency in $d$ in \cref{thm:final-ghazi} is a \emph{direct consequence} of the additive error of the DP coreset having the same dependency in \cite{GhaziKM20}. 
\end{remark}

\subsection{Related Work}
In the offline setting, private clustering was first studied by \cite{GuptaLMRT10}, and~\cite{FeldmanFKN09}. In particular,~\cite{GuptaLMRT10} used the exponential mechanism to produce a pure DP polynomial-time algorithm that achieves $(O(1),\Tilde{O}(k^2\Lambda))$-approximation in discrete spaces. However, their algorithm is highly inefficient in Euclidean space (see~\cite{StemmerK18} for a detailed exposition).
\cite{BalcanDLMZ17, StemmerK18} focused on designing an efficient polynomial time algorithm for clustering that achieves a constant (multiplicative) factor approximation in high-dimensional Euclidean space by adopting the techniques of \cite{GuptaLMRT10} while maintaining efficiency. \cite{FeldmanFKN09} introduced the notion of private coresets and a recent line of work has adopted their techniques to give clustering algorithms with better approximation guarantees and efficiency~\cite{FeldmanXZR17, NissimS18}. \cite{StemmerK18} gave a pure DP $k$-means algorithm with $O(1)$-multiplicative error, \cite{ChaturvediNX21} achieved the same multiplicative error but improved the additive error to $\Tilde{O}(k \sqrt{d}\Lambda^2/\eps)$. \cite{Nguyen20} and \cite{GhaziKM20} gave pure DP clustering algorithms that achieve optimal multiplicative error of $(1+\gamma)w^*$, where $w^*$ is the best approximation ratio for non-private $k$-means clustering.

\paragraph{Concurrent Works.} \cite{latour2023differential} concurrently released results on differentially private $k$-means clustering in the continual release setting. 
 Their work is mostly incomparable with ours --- their main focus is to optimize the approximation ratio under the more general setting of continual observation with the insertion \emph{and} deletion of points, but their algorithm does not optimize for space and requires storing all the input points in the stream. More recently~\cite{TourHS24} showed that a twenty year old greedy algorithm for clustering can be made DP and gives DP clustering algorithms for both the static and continual release setting that achieves the same multiplicative error as \cite{GhaziKM20} and additive error that has a $\log^{1.5}(T)$ dependency. This algorithm also does not optimize for space and inherently stores all the points in the stream.

 In contrast, we work in the insertion-only streaming setting and our main focus is to optimize the approximation ratio using low space, i.e., $\poly\log(T)$, in the continual release setting. Note that we achieve the same multiplicative error as~\cite{latour2023differential} in \cref{thm:final-ghazi}.

\subsection{Our Techniques}
Our techniques apply to both $k$-means and $k$-median clustering, but for simplicity, we explain our techniques for the problem of $k$-means. We first outline the challenges to designing a DP $k$-means clustering algorithm in the continual release setting and then discuss the main ideas behind our DP clustering framework. 

\paragraph{Naive Merge and Reduce approaches fail. }
A standard streaming technique for clustering is the Merge-and-Reduce framework~\cite{har2004coresets,AgarwalHV04,feldman2020turning}. The idea is to maintain coresets $C_1,\dots,C_L$ for $L=\log T$ levels (using a static coreset algorithm as a black-box) such that their union forms a coreset of all points seen so far.~\cite{cohen2022towards} showed that any point set has a $(1+\gamma)$-coreset with size $k\cdot \poly(1/\gamma)$, thus each level has capacity threshold $M = k \cdot \poly(\log T)$. When a new point $x$ arrives, it is inserted into $C_1$. For $i\in [L]$, if some $C_i$ exceeds size $M$, we compress it into a $(1+1/\poly(\log T))$-approximate coreset $C'$, reset $C_i \gets \emptyset$, and insert $C'$ into $C_{i+1}$. This process repeats until all levels respect the size bound. By induction, each $C_i$ is a $(1+1/\poly(\log T))^i$-coreset of its underlying input points, and hence $C_1\cup C_2\cup \cdots \cup C_L$ is a $(1+1/\poly(\log(T)))^L=(1+O(\log(T))/\poly(\log(T)))$-coreset of the entire input stream.

We can adapt the classical merge-and-reduce method to the DP setting as follows. At level $1$, whenever the buffer $C_1$ reaches size $M$, we run a DP $k$-means coreset routine to compress it. The resulting summary is inserted into $C_2$, and higher levels merge summaries using a non-private coreset algorithm. By the post-processing property, this preserves differential privacy while avoiding repeated privacy loss at higher levels. Note that the size checks that trigger compression at each level can be implemented privately with the Sparse Vector Technique (SVT)~\cite{DworkNRRV09, HardtR10} (see \cref{sec:dp-m-r} for details). 

The key point is that any DP coreset must introduce some additive error $\eta=\Omega(\Lambda^2)$~\cite{GuptaLMRT10,FeldmanFKN09}. If we apply this framework to the entire stream, the additive error introduced at the base still accumulates across levels: each merge adds the error from its children. After $\log T$ levels, this leads to a total additive error of $T^{\Omega(1)}\cdot \Lambda^2$, which is too large to yield meaningful guarantees. 

\paragraph{Our Approach. }Our main technical innovation is to design methods that prevent the additive error from growing exponentially, as in the naive implementation of the DP merge-and-reduce framework described above.

Rather than applying a single instance of the DP merge-and-reduce framework to the entire input stream, we first partition the space $\mathbb{R}^d$ so that nearby input points are grouped together, and apply multiple instances of the DP merge-and-reduce framework to disjoint groups in parallel. To do this, we compute a set of candidate centers forming a \emph{bicriteria approximation}\footnote{Here, a bicriteria approximation means that the algorithm may output more than $k$ centers while still approximating the $k$-means cost.} for $k$-means, and then assign input points to groups based on these centers. Within each group, we run a DP merge-and-reduce framework: the first level uses a DP clustering or DP coreset algorithm, while higher levels are computed using a non-DP coreset algorithm. This produces a DP \emph{semicoreset} (see \cref{def:semicoreset}). The semicoreset abstraction makes our framework flexible, since it can incorporate either a DP clustering algorithm or a DP coreset algorithm as the base primitive in the merge-and-reduce process.

For every timestamp $t \in [T]$, on a high-level, our algorithm (\cref{alg:extend-cluster}) does the following --- 
\begin{enumerate}
    \item Compute a set of centers $\cF$ in an online DP fashion that satisfies a bicriteria approximation (\cref{alg:dp-find-center-online}) to $k$-means.
    \item Maintain DP semicoresets of the points assigned to centers in $\cF$ in parallel via multiple DP Merge-and-Reduce (\cref{alg:dp-merge-reduce}) instances
    \item Output the union of these semicoresets called $\hat{\cY}$. 
\end{enumerate}
In a post-processing step --- Compute a non-DP $k$-means $(1+\gamma)$-approximation algorithm on the output $\hat{\cY}$. 

In the remainder of this section, we briefly discuss the different components of our approach.

\paragraph{Bicriteria Approximation. }\sloppy We first design a sublinear space online $\eps$-DP algorithm that outputs a set of $\Tilde{O}(k\log(T))$ candidate centers satisfying a $d^{O(1)}$-multiplicative approximation and $ \tilde{O}(\frac{d^2 \Lambda^2 k}{\eps}\cdot \poly\log (T))$ additive error. The bicriteria approximation algorithm uses two main ingredients --- \emph{quadtrees} and \emph{heavy hitters}. A quadtree creates a nested series of grids that partition $\mathbb{R}^d$ and can be used to embed input points into a Hierarchically Separated Tree (HST) metric, which often simplifies the analysis of $k$-means cost. We initialize our quadtree with $\log(\Lambda)$ levels, and use this embedding to map every input point to the center of a grid cell at each level. 

For a fixed level, our goal is to approximately choose $\tilde{O}(k)$ cells that have the most points and store them as candidate centers in set $\cF$ while using low space. To do this, we first hash cells of the fixed quadtree level to $w=O(k)$ buckets, and at every timestep, we track the DP size of each bucket through the standard Binary Mechanism~\cite{DworkNPR10, ChanSS11}, as well as privately compute the heavy hitters of each bucket using the DP heavy hitter algorithm from~\cite{EMMMVZ23} . Note that we need to track the DP size of each bucket using the Binary Mechanism separately to ensure that we can prune false positives among candidate centers. Without this pruning step, cells with very small counts could be mistakenly retained as heavy hitters. To the best of our knowledge this is the first use of DP Heavy hitters for DP clustering \emph{in the continual release setting}.

Finally, the algorithm outputs the cumulative set of candidate centers $\cF$ across each of the $\log \Lambda$ levels of the quadtree, yielding a bicriteria approximation. See \cref{thm:final-acc-dp-hh} for a formal statement.

\paragraph{Grouping points and applying Merge and Reduce. }
Instead of running a single instance of the DP Merge and Reduce framework over the entire stream (as described in the naive approach), our strategy is to partition the space $\bbR^d$ into groups such that input points close to a candidate center (of the bicriteria solution $\cF$) are in the same group and then run $\log(\Lambda)$ parallel instances of the DP Merge and Reduce framework \emph{per group}. We use a partitioning technique introduced by~\cite{Chen09} and define \emph{rings} $R_r$ (where $1\leq r \leq \log \Lambda$) based on their distance from the current centers in $\cF$ and map input points to these rings (see \cref{def:ring-center} for a formal definition).

The intuition is that any DP coreset necessarily incurs additive error that scales with the square of the diameter of the set it is applied to. If we ran a single DP Merge and Reduce on all points assigned to a center, the diameter could be as large as $\Lambda$, leading to $\Omega(\Lambda^2)$ additive error. By splitting the group into $\log(\Lambda)$ concentric rings, each ring has diameter only $O(2^r)$, so the additive error of the DP routine on ring $R_r$ is $O(2^{2r})$. This error is now proportional to the cost contributed by points in that ring, which allows us to keep the total error under control.

Now we can run the DP Merge and Reduce framework and obtain DP semicoresets for points in each ring, and by taking the union of these semicoresets over all rings, we obtain a DP semicoreset for the stream of points seen so far. Since the DP clustering algorithms we use as a black-box in our Merge and Reduce framework from~\cite{GhaziKM20,StemmerK18} achieve constant multiplicative error, this technique will result in constant multiplicative error of the resulting DP semicoreset.

\paragraph{Charging Additive Error to Multiplicative Error. }
Because we are taking the union of DP semicoresets over all rings, the additive error resulting from this union depends on $\sum_r N_r/M$, where $N_r$ is the number of points in ring $r$ and $M$ controls when the DP Merge and Reduce framework is applied. At first glance, this sum could be as large as $T/M$, which seems too large to be useful.

The crucial insight is that we use a structural property of the bicriteria solution\footnote{Intuitively, this property holds because the bicriteria solution already ``pays'' for the squared distances that define the rings: every point is assigned to a nearby candidate center, and the radius of its ring is proportional to this distance. } (\cref{lem:radius-opt-cost}): the total sum of squared ring radii over all points is bounded, up to constants, by the $k$-means cost of the bicriteria solution $F$ (i.e., $\text{cost}(F,S) = \sum_{x\in S} d(x,F)^2$), which in turn is within a constant factor of the optimal $k$-means cost plus a small additive term. This property ensures that the additive error accumulated by the DP Merge and Reduce instances is controlled by the clustering cost. In turn, this allows us to charge the additive error to the multiplicative error, so that with an appropriate choice of $M$ the overall error is just a small constant-factor blow-up in the approximation plus a polylogarithmic additive term (see \cref{app:framework-accuracy} for details).

Finally, observe that the bicriteria solution $\cF$ can change over time, i.e., new centers can be added to $\cF$. Thus whenever a new center (or centers) are added to $\cF$, we redefine the rings (based on distance to the set $\cF$) and initiate new DP Merge and Reduce instances for each ring.

\paragraph{Comparison of our techniques to~\cite{EZNMC22}. }\cite{EZNMC22} gave DP $k$-means and $k$-median algorithms in the Massively Parallel Computation (MPC) model which achieved $O(1)$-multiplicative and $\poly(k,d)$ additive error. Although their techniques are similar to ours, in the sense that they reduce the problem of solving DP clustering on the entire dataset to solving multiple instances of the DP clustering problem on partitioned input data, there are several differences in our work and new challenges that we have to deal with as a result.

First, their work does not tackle sublinear space in the continual release setting. 
Second, our bicriteria approximation algorithm uses a black-box continual release \emph{DP heavy hitters} algorithm from~\cite{EMMMVZ23} to find the candidate centers and thus requires a completely different analysis compared to the greedy approach used in~\cite{EZNMC22}. More importantly, we have to implement a novel DP Merge and Reduce framework that works with both an offline DP clustering or DP coreset algorithm, which was not required in the MPC setting.

\section{Differentially Private Clustering Framework}\label{sec:dp-framework}
In this section, we first present our differentially private clustering framework in the continual release setting which is given by \cref{alg:extend-cluster} and discuss the different components of the algorithm as well as its analysis in detail (we defer some of the proofs to \cref{app:dp-framework-proofs}). Notably, our framework allows us to plug in any existing offline DP coreset or clustering algorithm to obtain a corresponding DP clustering algorithm in the continual release setting. The proofs of \cref{thm:main-inf}, \cref{thm:final-stemmer} and \cref{thm:final-ghazi} are stated in \cref{sec:main-proofs}.  

\begin{algorithm}[!htb]
\caption{Main Algorithm}
\label{alg:extend-cluster}

\KwData{ Stream $\cS$ of points $x_1, \ldots, x_T \in \bbR^d$, Privacy parameters $\eps,\eps_1,\delta_1$}

Call $\textbf{Initialize}(\eps)$ of $\mathsf{\textbf{BicriteriaDPCenters}}$\;\Comment{See \cref{alg:dp-cluster-online}}

Initialize bicriteria solution $\cF=\emptyset$\;
Initialize flag for new centers being added $flag_{new}=0$, and DP semicoreset $\hat{\cY}=\emptyset $\;

\For{new point $x_t$}{
$\cF_t \leftarrow $\textbf{Update}$(x_t)$ of $\mathsf{\textbf{BicriteriaDPCenters}}$\;

\If{$\cF_t \neq \emptyset$ and $\vert \cF \cap \cF_t \vert < \vert \cF_t \vert $}{ \Comment{New centers added to $\cF$ --- need to redefine rings}

$flag_{new}=1$;

}
$\cF \leftarrow \cF \cup \cF_t$\;

\For{$1 \leq r \leq \log \Lambda$, run the following in parallel}{%\label{b1}
Let $R_r$ represent the ring centered at $\cF$ (see \cref{def:ring-center})\;

\If{$flag_{new}=1$}{
Create new instance $\textbf{DP-Merge-Reduce}_{R_r}$ by calling $\textbf{Initialize}(\eps,\eps_1,\delta_1)$ of $\textbf{DP-Merge-Reduce}$\;\Comment{See \cref{alg:dp-merge-reduce}}
 }%\EndIf 
 
\eIf{$ x_t \in R_r$ }{
$\hat{\cY}_{r} \leftarrow$ \textbf{Update}$(x_t)$ of $\textbf{DP-Merge-Reduce}_{R_r}$\; }
{
$\hat{\cY}_{r} \leftarrow$ \textbf{Update}$(\perp)$ of $\textbf{DP-Merge-Reduce}_{R_r}$\;
}}
$flag_{new}=0$\;

{Set $\hat{\cY} \leftarrow \hat{\cY} \cup (\cup_{r} \hat{\cY}_{r})$\;} 

{Run non-DP $(1+\gamma$)-coreset algorithm on $\hat{\cY}$ and store this semicoreset as the new $\hat{\cY}$\;}%\label{b2}

Output the semicoreset $\hat{\cY}$\;

{Delete existing input points (if any) in memory\;}
}
\end{algorithm}

Before presenting the main algorithm, we formally define a ring centered at a set as we will use this to partition the space $\bbR^d$ in our algorithm. 
\begin{definition}[Ring centered at a Set]\label{def:ring-center}
Let $r \in \bbR$. Ring $R_r$ for set $\cF$ contains the set of points $\{x_i\}_{i \in [T]}$ such that $2^{r-1} \leq d(x_i,\cF)< 2^{r}$.
\end{definition}

\paragraph{Main Algorithm (\cref{alg:extend-cluster}). }When a new point $x_t$ arrives, our algorithm does the following
\begin{enumerate}
    \item \label{bicriteria} Update the bicriteria solution $\cF$ (see \cref{alg:dp-cluster-online})
    \item If new centers have been added to $\cF$ then create rings according to \cref{def:ring-center} for the set $\cF$ and add $x_t$ to a  ring\footnote{Notice that in the pseudo code \cref{alg:extend-cluster} the symbol $\perp$ represents an empty update that is effectively ignored. This is needed for technical reasons to ensure DP by avoiding the value of the input affecting the number of events in the sub-streams.}. For each ring $1 \leq r \leq \log(\Lambda)$ create an instance of $\textsf{DP-Merge-Reduce}_r$ (see \cref{alg:dp-merge-reduce}) which outputs a DP semicoreset per ring. \\If no new centers have been added to $\cF$ in this timestep, then instead of creating new rings, the algorithm adds $x_t$ to an existing ring (and corresponding \textsf{DP-Merge-Reduce} instance).
    \item Release the union of these DP semicoresets as $\hat{\cY}$. 
\end{enumerate}

In order to keep our space usage small, we apply a $(1+\gamma)$-approximation non-DP coreset algorithm to the union of semicoresets of these rings as $\hat{\cY}$.

Finally, in an offline postprocessing step, we apply a $\rho$-approximate non-DP clustering algorithm to $\hat{\cY}$.

\paragraph{Analysis. }
For the sake of analysis, we split the entire stream of points into \emph{epochs} dictated by the addition of \emph{new} centers to the bicriteria solution $\cF$. 
Let $T_1, \ldots, T_e$ be the epochs such that for a fixed $i$, the set of bicriteria centers $\cF$ over the timesteps $t \in T_i$ is fixed. Clearly $T_1 \cup \ldots \cup T_e = [T]$ and $T_1 \cap \ldots \cap T_e = \emptyset$.

We first state the theoretical guarantees of \textsf{BicriteriaDPCenters} and \textsf{DP-Merge-Reduce} as we need these results to state the guarantees of our Main Algorithm (see \cref{alg:extend-cluster}). The proofs for these statements can be found in \cref{sec:final-acc-dp-hh} and \cref{sec:dp-merge-reduce}.  The accuracy and space analysis for \cref{alg:extend-cluster} can be found in \cref{app:dp-framework-proofs}.

\begin{theorem} \label{thm:final-acc-dp-hh}[\textsf{BicriteriaDPCenters}]
Let $\cS:= \{x_1,\ldots, x_T\}$ be the stream of input points in Euclidean space. For $t\in [T]$, let $\cF_t$ be the set of centers until time step $t$. Let $cost(\cF, \cS):=\sum^T_{t=1} cost(\cF_t)$ where $cost(\cF_t) := \min_{f \in \cF_t} dist^2(x_t,f)$. There exists an algorithm \textsf{BicriteriaDPCenters} (see \cref{alg:dp-cluster-online}) that outputs a set of centers $\cF$ at every timestep $t\in [T]$ such that 
\begin{enumerate}
    \item(Privacy) \textsf{BicriteriaDPCenters} is $\eps$-DP. 
    \item(Accuracy) With probability at least  $1 - k^{-\operatorname{poly}(k,\log \Lambda)}$, 
    \begin{align*}
     cost(\cF, \cS) &\leq  O(d^3)cost(C^{opt}_{\cS}, \cS) + \tilde{O}\left(\frac{d^{2} \Lambda^2 k}{\eps} \poly\left(\log\left({T\cdot k \cdot\Lambda}\right)\right)\right)
\end{align*}
where $cost(C^{opt}_{\cS}, \cS)$ is the optimal $k$-means cost for $\cS$.

\item(Space) \textsf{BicriteriaDPCenters} uses $\tilde{O}(k \poly\left( \log \left({T}{\Lambda}k\right)\right))$ space.

\item (Size) $\cF$ has at most $\tilde{O}({k}\log T)$
centers.
\end{enumerate} 
\end{theorem}

Recall that our main algorithm runs multiple concurrent instances of \textsf{DP-Merge-Reduce} at any given timestep, so $N$ in the theorem statement below represents the number of points seen by a specific instance of \textsf{DP-Merge-Reduce}. We use a non-DP coreset algorithm in \textsf{DP-Merge-Reduce} which we fix in the analysis to be the current state-of-the-art clustering result by\cite{Cohen-AddadSS21} (see formal guarantees in~\cref{thm:nondp-coreset}).

\begin{theorem}[\textsf{DP-Merge-Reduce}] \label{thm:dp-merge-reduce}
Let $0<\xi<1$, $T$ be the length of the entire stream, $\eps,\delta$ be privacy parameters, $M$ be an arbitrary parameter such that  $M>\frac{12}{\eps}\log(\frac{2T}{\xi})$, and $P$ be a sub-stream of non-empty points with length $N$. 

Suppose we are given black-box access to an offline $(\eps_1,\delta_1)$-DP algorithm $\cA$ that computes a $(\kappa,\eta_1,\eta_2)$-semicoreset of $X \subseteq \bbR^d$ of size $SZ_\cA(N,k,d,\eps, \delta, \kappa,\eta_1,\eta_2, \xi_A)$ using space $S_\cA(N,k,d,\eps, \delta, \kappa,\eta_1,\eta_2, \xi_A)$ with failure probability $\xi_A$. And we are given black-box access to an offline non-DP algorithm $\cB$ that computes a $(1+\gamma)$-coreset of $X \subseteq \bbR^d$ of size $SZ_{\cB}(N,k,d,\gamma,\xi_B)$ using space $S_\cB(N,k,d,\gamma, \xi_B)$ with failure probability $\xi_B$. 
Then there exists an algorithm \textsf{DP-Merge-Reduce} (see \cref{alg:dp-merge-reduce}) in the streaming model such that
\begin{itemize}
    \item (Privacy) \textsf{DP-Merge-Reduce} is $(\eps+\eps_1,\delta_1)$-DP. %under the continual release setting.
    \item (Accuracy) With probability $1-\xi_A-\xi_B-\xi$, the semicoreset released by \textsf{DP-Merge-Reduce} is  a $((1+\gamma)\kappa,(\frac{4N}{M}-1)(1+\gamma)\eta+\tilde{M})$-semicoreset of $P$. Where $\tilde{M}:= {M}+\frac{6}{\eps}\log(\frac{2T}{\xi})$.
    \item (Space) \textsf{DP-Merge-Reduce} requires 
    
    \sloppy$S_\cA(M,k,d,\eps, \delta, \kappa,\eta_1,\eta_2, \xi_A)+ S_\cB(SZ_\cA(M,k,d,\eps, \delta, \kappa,\eta_1,\eta_2, \xi_A),k,d,\gamma, \xi_B)+\lceil \log(2N/M) \rceil \cdot S_\cB(SZ_\cB(M,k,d,\gamma,  \xi_B),k,d,\gamma, \xi_B) +3M/2$ space.

    \item (Size) The resulting coreset has size at most  $O(k\log k \cdot \gamma^{-4})$.
\end{itemize}
\end{theorem}

\begin{remark}\label{rem:param-M}
The parameter $M$ denotes the block size of the base level in \textsf{DP-Merge-Reduce} (see \cref{alg:dp-merge-reduce}). We treat $M$ as an arbitrary parameter in \cref{thm:dp-merge-reduce}, \cref{lem:main-space} and \cref{lem:yhat-dp-merge-reduce-output-semicoreset}. However we set $M$ to be a function of $\alpha$ (the multiplicative approximation error of the bicriteria solution), $\eta_2$ (the additive approximation error of the DP semicoreset algorithm) and a parameter $C_M$ in the proof of \cref{thm:yhat-semicoreset}, i.e., $ M:= \frac{\alpha \eta_2 }{C_M}$. To obtain an $O(1)$-multiplicative approximation in \cref{thm:final-stemmer}, we assign $C_M=O(1)$. To obtain 
a $(1+\gamma')$-multiplicative approximation in \cref{thm:final-ghazi}, we assign $C_M=\gamma/2$. See proof of \cref{thm:dp-clustering-stemmer} and \cref{thm:dp-coreset-ghazi} for details.
\end{remark}

\paragraph{Privacy. }We first show that the output of our main algorithm is indeed differentially private.  
\begin{lemma}\label{lem:dp-main}
Let the underlying DP semicoreset algorithm used as a blackbox be $(\eps_1,\delta_1)$-DP, then \cref{alg:extend-cluster} is $(\eps_1+2\eps,\delta_1)$-DP under the continual release setting.
\end{lemma}
\begin{proof}
We first observe that for a fixed input point $x_t$, we can view \cref{alg:extend-cluster} as two main steps (1) $x_t$ is first used to update the subroutine \textsf{BicriteriaDPCenters} which produces a set of candidate centers $\cF$ (2) $x_t$ is assigned to an appropriate ring $R_r$ (defined according to $\cF$) and added to the corresponding \textsf{DP-Merge-Reduce}$_{R_r}$ instance of that ring which produces a DP semicoreset (3) Finally, the algorithm releases the union of these DP semicoresets over disjoint rings at timestep $t$. 

Step (1) involves $x_t$ being processed by \textsf{BicriteriaDPCenters} which is $\eps$-DP under the continual release setting (by \cref{thm:final-acc-dp-hh}). In Step (2), since the rings partition the input space by definition, the corresponding \textsf{DP-Merge-Reduce}$_{R_r}$ instances are disjoint. Since the semicoreset output by each \textsf{DP-Merge-Reduce}$_{R_r}$ instance is $(\eps_1+\eps,\delta_1)$-DP under the continual release setting (by \cref{thm:dp-merge-reduce}), then by parallel composition (Item \ref{parallel} of \cref{thm:composition}), the total privacy over all semicoresets output by \textsf{DP-Merge-Reduce}$_{R_r}$ instances (where $1 \leq r \leq \Lambda$) is $(\eps_1+\eps,\delta_1)$-DP. In Step (3), the union of DP semicoresets is DP by postprocessing. 

Thus, each step is private under the continual release setting and note that each step is carried out in the same fixed timestep $t$.  By sequential composition (Item \ref{sequential} of \cref{thm:composition}) of the privacy for Steps (1) and (2), and postprocessing of Step (3), our claim follows. 
\end{proof}

\subsection{Proofs of \cref{thm:main-inf}, \cref{thm:final-stemmer} and \cref{thm:final-ghazi} }\label{sec:main-proofs}
We now present the proof of the guarantees of our DP clustering framework (\cref{alg:extend-cluster}) stated in \cref{thm:main-inf} below. 
\begin{proof}%[Proof of \theoremref{main-inf}]
The privacy guarantee follows from \cref{lem:dp-main}, and the accuracy guarantee follows from \cref{thm:final-accuracy-main}. The space guarantee follows from \cref{lem:main-space}, where we simplify the second term $S_\cB (SZ_\cA(\cdot))$ in the statement of \cref{lem:main-space} by noting that the space used by the non-DP coreset algorithm $\cB$ is linear in the semicoreset size of the algorithm $\cA$. 
\end{proof}

Next we present the proofs of \cref{thm:final-stemmer} and \cref{thm:final-ghazi}. We first restate the existing offline DP clustering and coreset results that we use as a blackbox from \cite{StemmerK18} and \cite{GhaziKM20} in \cref{thm:dp-clustering-stemmer} and \cref{thm:dp-coreset-ghazi}. Note that in~\cref{thm:final-stemmer}, we need to convert the set of $k$ centers obtained from applying the DP clustering algorithm in~\cref{thm:dp-clustering-stemmer} into a semicoreset via~\cref{thm:clustering-to-semicoreset}. 
\begin{theorem}[Theorem 3.10 in~\cite{StemmerK18}]\label{thm:dp-clustering-stemmer}
There is an $(\eps,\delta)$-DP algorithm that takes a database $S$ containing $n$ points in the $d$-dimensional ball $\cB(0,\Lambda)$, and outputs with probability $1-\beta$, a $(\gamma,\eta)$-approximation for the $k$-means of $S$ where $\gamma=O(1)$ and $\eta={O}((k^{1.01} \cdot d^{0.51} + k^{1.5})\cdot \poly(\log(n),\log(\frac{1}{\beta}),\log(\frac{1}{\delta}),\frac{1}{\eps}) )$
\end{theorem}

The proof of \cref{thm:final-stemmer} is stated below. 

\begin{proof}
We first note that using \cref{thm:clustering-to-semicoreset} together with \cref{thm:dp-clustering-stemmer} gives us an algorithm $\cA$ that produces a $(\kappa,\eta_1,\eta_2)$-semicoreset where $\kappa=O(1), \eta_1=O(1)$ and $\eta_2={O}((k^{1.01} \cdot d^{0.51} + k^{1.5}+k)\cdot \poly(\log(T) \log(\frac{1}{\delta}),\frac{1}{\eps}) )$ and that we can use as a black-box in our DP clustering framework. 
The privacy guarantee follows by observing that the resulting algorithm from \cref{thm:clustering-to-semicoreset} is $(3\eps,\delta)$-DP and applying Item~\ref{it:dp-main} of \cref{thm:main-inf}. %using \lemmaref{dp-main}. 
We choose $C_M=100$ and the accuracy guarantee follows from Item~\ref{it:acc-main} of \cref{thm:main-inf}. 

\sloppy Finally, for the space usage, we first determine the value of $M$ in \cref{thm:main-inf}. Because $\eta_2={O}((k^{1.01} \cdot d^{0.51} + k^{1.5})\cdot \poly(\log(T) \log(\frac{1}{\delta}),\frac{1}{\eps}) )$ and $C_M=100$, we have that $M= {O}((k^{1.01} \cdot d^{3.51} + d^3 \cdot k^{1.5})\cdot \poly(\log(T), \log(\frac{1}{\delta}),\frac{1}{\eps}) )$. 

Note that $S_\cA(M,k,d)= O((M+k)d)+O(k^{1+\eta} \cdot d \cdot \log(M))$ for some arbitrarily chosen constant $\eta>0$ where the second term comes from the space usage of~\cref{thm:clustering-to-semicoreset}. Also $SZ_\cA(\cdot) = k$. %and $S_\cB (SZ_\cA(\cdot))= S_\cB (k) =\tilde{O}(k)$. 
The final space complexity is obtained by plugging in the value of $M$ and setting $\eta=0.01$.  
\end{proof}

\begin{theorem}[Lemma 16 in~\cite{GhaziKM20}]\label{thm:dp-coreset-ghazi}
There is a $2^{O_\gamma(d)}\poly(n)$-time $\eps$-DP algorithm that with probability 0.99 produces an $\left(1+\gamma, O_\gamma\left(\frac{k^2 2^{O_\gamma(d)}}{\eps} \poly\log n\right)\right)$-coreset for $k$-means (and $k$-median). The size of the coreset is $2^{O_\gamma(d)} \cdot \poly(k,\log n)$. 
\end{theorem}

The proof of \cref{thm:final-ghazi} is stated below. 
\begin{proof}
First note that the $\eps$-DP coreset from \cref{thm:dp-coreset-ghazi}~\cite{GhaziKM20} is a $(\kappa,0,\eta)$-semicoreset where $\kappa=1+\gamma$ and $\eta = O_\gamma\left(\frac{k^2 2^{O_\gamma(d)}}{\eps} \poly\log (T)\right)$ and thus we can use it directly as a blackbox in our DP clustering framework. The privacy guarantee follows from Item~\ref{it:dp-main} of \cref{thm:main-inf}. The accuracy guarantee follows from Item~\ref{it:acc-main} of \cref{thm:main-inf} where we choose $C_M=\gamma/2$. 

Finally, for the space usage, we first note that in \cref{thm:main-inf}, by plugging in the values of $C_M$ and $\eta$, we have that $M=  O_\gamma\left(\frac{d^3k^2 2^{O_\gamma(d)}}{\eps} \poly\log T \right)$. By observing that $S_\cA(M,d,k)= \poly(M,d,k)$ and total coreset size $SZ_\cA(\cdot) =2^{O_\gamma(d)} \cdot \poly(k,\log T)$ and plugging in the value of $M$ in Item~\ref{it:space-main} of \cref{thm:main-inf}, the space guarantee follows.
\end{proof}

\section{Conclusion and Future Work}
To the best of our knowledge, in this paper we designed the \emph{first sublinear space differentially private} algorithms for $k$-clustering in the continual release setting (sublinear in stream length $T$). Our paper gives a general DP clustering framework that uses any existing offline DP clustering or coreset algorithm as a black-box and the space guarantee of the final algorithm depends on the additive error of the offline DP algorithm. In particular, we achieve a $(1+\gamma)$-multiplicative approximation using the DP coreset algorithm from~\cite{GhaziKM20} with $\tilde{O}_\gamma(\poly(k,2^{O_\gamma(d)},\log(T)))$ space and an $O(1)$-multiplicative approximation with $\tilde{O}(k^{1.5} \cdot \poly(d,\log(T)))$ space using the DP $k$-means algorithm from~\cite{StemmerK18}.  

While our work achieves sublinear space in the insertion-only streaming setting, recently~\cite{latour2023differential} studied the DP clustering problem under insertions and deletions but their results do not optimize for space --- an exciting open direction is to design sublinear space differentially private algorithms for $k$-clustering in this model. Finally proving information-theoretic space lower bounds for any DP problem in the streaming setting is another exciting open direction. The only known space lower bound~\cite{DinurSWZ23} is under cryptographic assumptions, and does not apply to clustering problems.

\bibliographystyle{plain}
\bibliography{references}

\appendix 
\section{Preliminaries}

\paragraph{Norms and heavy hitters. }Let $p \geq 1$, the $\ell_p$-norm of a vector $\mathbf{x}=(x_1, \ldots, x_t)$ is defined as $\|\mathbf{x}\|_p = (\sum^t_{i=1} \vert x_i \vert^p)^{1/p}$. Given a multiset $\cS$, denote the frequency of an item $x$ appearing in $\cS$ as $f(x)$. We say that an item $x$ is an $\theta$-heavy hitter ($\theta$-HH for short) if $f(x) \geq \theta \|\cS \|_1$.

\paragraph{Differential Privacy background. }
We state some DP theorems that we frequently utilize below. For a more detailed background on DP refer to~\cite{DworkR14}.
\begin{theorem}[Composition Theorems~\cite{mcsherry2009privacy}]\label{thm:composition}
\begin{enumerate}
    \item \label{sequential} (Sequential) Let $M_i$ each provide $(\eps_i,\delta_i)$-DP. The sequence of $M_i(X)$ provides $(\sum_i\eps_i, \sum_i \delta_i)$-DP.
    \item \label{parallel}(Parallel) Let $M_i$ each provide $(\eps,\delta)$-DP. Let $D_i$ be arbitrary disjoint subsets of input domain $D$. The sequence of $M_i(X \cap D_i)$ provides $(\eps,\delta)$-DP. 
\end{enumerate}
\end{theorem}

\begin{theorem}[Binary Mechanism \textsf{BM}~\cite{ChanSS11, DworkNPR10}]\label{thm:binarymechanism}
Let $\eps \geq 0, \gamma \in (0,0.5)$, there is an $\eps$-DP algorithm for the sum of the stream in the  continual release model. With probability $1-\xi$, the additive error of the output for every timestamp $t \in [T]$ is always at most $O(\frac{1}{\eps} \log^{2.5}(T) \log(\frac{1}{\xi}))$ and uses $O(\log T)$ space.
\end{theorem}

\begin{theorem}[\textsf{DP-HH} algorithm~\cite{EMMMVZ23}]\label{thm:dp-hh-cr} Let $\eps>0$, $\gamma_h \in (0,0.5)$, $0< \theta < 1$, $\xi \in (0,0.5)$. There is an $\eps$-DP algorithm in the streaming continual release model such that with probability at least $1-\xi$, it always outputs a set $H \subseteq \cU$ and a function $\hat{f} : H \to \bbR$ for every timestamp $t \in [T]$ such that 
\begin{enumerate}
    \item $\forall a \in H$, $\hat{f}(a) \in (1\pm \gamma_h) \cdot f_a$ where $f_a$ is the frequency of $a$ in the stream $\cS = (a_1, a_2, \ldots, a_t)$
    \item $\forall a \in \cU$, if $f_a \geq \frac{1}{\eps \gamma_h} \poly\bigl(\log\bigl( \frac{T \cdot\vert \cU \vert }{\theta \xi \gamma_h }\bigr)\bigr)$ and $f^1_a \geq \theta  \|\cS \|_1$ then $a \in H$
    \item The size of $H$ is at most $O((\log(T/\xi)+\log \vert \cU \vert ) \cdot (\frac{1+\gamma_h}{1-\gamma_h}) \cdot \frac{1}{\theta})$
\end{enumerate}
 The algorithm uses $\frac{1 }{\gamma_h^2\theta^3} \poly \left( \log\left( \frac{T \cdot  \vert \cU \vert}{\xi \theta }\right)\right)$ space. 
\end{theorem}

\begin{remark}
In \cref{alg:dp-find-center-online} we use the notation \textbf{Initialize} and \textbf{Update} for the Binary Mechanism (\textsf{BM}) and DP-Heavy Hitters (\textsf{DP-HH}) primitives. For clarity:
\begin{itemize}
    \item \textbf{Initialize} sets up the respective data structure with its privacy budget and parameters (e.g., stream length, thresholds).
    \item \textbf{Update} processes a single incoming item (or the empty symbol $\perp$) and returns the current private estimate. For BM this is a differentially private count, while for DP-HH it is a set of candidate heavy hitters together with approximate frequencies.
\end{itemize}
We do not re-write these algorithms explicitly, since our use is a direct adaptation of the continual release algorithms already stated in \cref{thm:binarymechanism}, \cref{thm:dp-hh-cr} and in the cited references.
\end{remark}
\paragraph{Useful Clustering Theorems. } We present some useful clustering theorems that we frequently utilize below.

We clarify that the statement below concerns sequential algorithms in the MPC model. However, when the entire input resides on a single machine, a sequential algorithm is trivially static. Since our use case involves static algorithms in such a setting, we restate the lemma accordingly to reflect this context more accurately.
\begin{theorem}[Lemma C.1 in~\cite{EZNMC22}] \label{thm:clustering-to-semicoreset}
Let $\cA$ be an $(\eps,\delta)$-DP (static) algorithm that takes as input a dataset $P$ of size at least $\Omega(k \cdot \eps^{-1} \log(n))$ contained in a ball $\cB(0,\Lambda)$ (with known center) and outputs a set of $k$ centers $\cC$ contained in $\cB$ such that $cost(\cC,\cX ) \leq O(1) cost(\cC^{opt},\cX)+V(n,d,k,\eps,\delta) \cdot \Lambda^p$. Then for any fixed $\eta>0$ there exists a $(3\eps,\delta)$-DP (sequential) algorithm $\cB$ on the dataset $\cY$ that outputs a set $\hat{P} \subset B$ such that for any set $\cC$ of size at most $k$. 
\begin{align}
    cost(\cC,\hat{P}) &\leq O(1) \cdot \left( cost(\cC^{opt}_P,P) + cost(\cC,P)+(V(n,d,k,\eps,\delta)+k \eps^{-1} \log(n) )\cdot \Lambda^p \right) \\
    cost(\cC,P) &\leq O(1) \cdot \left( cost(\cC^{opt}_P,P) + cost(\cC,\hat{P}) +(V(n,d,k,\eps,\delta)+k \eps^{-1} \log(n))\cdot \Lambda^p \right)
\end{align}
In addition constructing the set requires only $O(k^{1+\eta} \cdot d \cdot \log(n))$ space and at most $O(k^\eta \cdot d \cdot \log(n))$ time per point in $P$ for some arbitrarily chosen constant $\eta>0$. 
\end{theorem}

\begin{theorem}[\cite{Cohen-AddadSS21}]%A New coreset Framework for Clustering
\label{thm:nondp-coreset}
There exists a non-DP $(1+\gamma)$-coreset of size $O(k\log k \cdot (\gamma^{-2-\max(2,z)}) \cdot 2^{O(z \log z)} \cdot \poly\log(\gamma^{-1}))$ for $(k,z)$-clustering in Euclidean spaces. 
\end{theorem}
\section{Bicriteria Approximation in Continual Release Setting}\label{sec:bicriteria}
We describe our bicriteria approximation algorithm and analysis in more detail here.

\begin{algorithm}[!htb]
\caption{\sf{BicriteriaDPCenters}}
\label{alg:dp-cluster-online}
%\begin{algorithmic}[1]
\KwData{Privacy parameter $\eps$, Stream $\cS$ of points $x_1, \ldots, x_T \in \bbR^d$ }
\vspace{0.75em}   

{\textbf{Initialize}($\eps$):}

{$\eps' := \frac{\eps}{4\log(\Lambda) \log^2(k)}$\;}

{Parallel quadtrees $Q_1,\ldots, Q_{\log(k)}$ such that:

 each quadtree $Q_q$ has $\log(\Lambda)$ levels with the bottom level having grid size $\Theta(1)$\;
 }
 
\For{$0 \leq \ell \leq \log(\Lambda)$ and $1 \leq q \leq \log(k)$}{

{Initialize $\mathbf{DPFindCenters}_{\ell,q}$ of $\mathbf{DPFindCenters}(\eps')$\; }
}
{Set of candidate centers $\cF := \emptyset$\;}
\vspace{0.75em}   
{\textbf{Update}($x_t$):}

%\For{when new point $x_t$ arrives}
\For{$0 \leq \ell \leq \log(\Lambda)$ and $1 \leq q \leq \log(k)$}{
%\State{Append $x_t$ to $\cS^{(q)}_\ell$}
{$\hat{\cF}_t \leftarrow \mathbf{Update}(x_t) \text{ of } \mathbf{DPFindCenters}_{\ell,q}$\;}\Comment{See \cref{alg:dp-find-center-online}}

{$\cF \leftarrow \cF \cup \hat{\cF}_t$\;}
}
Output $\cF$\;
%\EndFor
%\end{algorithmic}
\end{algorithm}

\begin{algorithm}[!htb]
\caption{\sf{DPFindCenters}}
\label{alg:dp-find-center-online}
%\begin{algorithmic}[1]
\KwData{Privacy parameter $\eps'$, 
Stream of points $x_1, \ldots, x_T$, level of quadtree $\ell$ 
}

{\textbf{Initialize}($\eps'$):}

{$\eps' \leftarrow \eps'$\;} \Comment{where $\eps' := \frac{\eps}{4\log(\Lambda) \log^2(k)}$}

{$w=O(k)$ \;}

{Hash function $h:[2^\ell] \to [w]$ s.t. $\forall \text{ cells } \bfc, \forall j \in [w], \Pr[h(\bfc)=j] = \frac{1}{w}$\;}

{$\hat{T}_{1}=0, \ldots, \hat{T}_{w}=0$\;}\Comment{DP Count for the size of hash bucket}

{Initialize $\textbf{BM}_1,\ldots,\textbf{BM}_w$ of $\textbf{BinaryMechanism}(T,\eps')$\;}\Comment{See \cref{thm:binarymechanism}~\cite{DworkNPR10}}

{Initialize $\textbf{DP-HH}_1,\ldots,\textbf{DP-HH}_w$ of $\textbf{DP-HH}(T,\eps')$\;}\Comment{See \cref{thm:dp-hh-cr}~\cite{EMMMVZ23}}

{\textbf{Update}($x_t$):}

{Initialize $\hat{\cF}_t=\emptyset$\;}

{Let $x_t$ be mapped to cell $\bfc^*$\;} \Comment{\textbf{DPFindCenters} is initialized per level $\ell$ of quadtree instance $q$}

\For{$p=1,\ldots,L$, where $L:=\log (k^2)$ run in parallel}{%\label{par-p}
\For{$j \in [w]$} {
\Comment{Update the DP count of each hash bucket}

\eIf{$j = h(\bfc^*)$}{ 
{$\hat{T}_{j} \leftarrow \textbf{Update}(1)\text{ of }\textbf{BM}_{j}$\;}
}{
{$\hat{T}_{j} \leftarrow \textbf{Update}(0)\text{ of }\textbf{BM}_{j}$}
}
}
\For{$j\in[w]$}{ 
\Comment{Update the DP HHs of each hash bucket}

\eIf{$h(\bfc^*)=j$}{
{$\hat{f}, H \leftarrow$ \textbf{Update}($\bfc^*$) of $\textbf{DP-HH}_{j}$\;}%\Comment{See \theoremref{dp-hh-cr}}
}{
{$\hat{f}, H \leftarrow$ \textbf{Update}($\perp$) of $\textbf{DP-HH}_j$\;}
}
\For{cell $\bfc \in H$}{
\If{$\hat{f}(\bfc) \geq \frac{\theta}{1000} \cdot \hat{T}_{h(\bfc)}$}{
{Add centerpoint of $\bfc$ to $\hat{\cF}_t$ as a center\;}\Comment{$\theta$ is the HH threshold parameter and set to be an appropriate constant }
}
}
}
}
{Return $\hat{\cF}_t$}\;
%\end{algorithmic}
\end{algorithm}

\paragraph{Algorithm. }Our bicriteria approximation algorithm is given by \cref{alg:dp-cluster-online} which initializes $\log(k)$ parallel instances of randomly shifted quadtrees. Each input point $x_t$ is assigned to a cell in every level of every quadtree. For a fixed quadtree $1 \leq q \leq \log(k)$ and fixed level $0 \leq \ell \leq \Lambda$, the subroutine \textsf{DPFindCenters} (see \cref{alg:dp-find-center-online}) returns a candidate set of centers $\hat{\cF}_t$ which is added to the current set of candidate set of centers $\cF$. 

The \textsf{DPFindCenters} subroutine (see \cref{alg:dp-find-center-online}) finds the approximate heaviest $O(k)$ cells in a fixed level of a fixed quadtree. It achieves this by first hashing the cell containing the current point to a bucket, note that there are $w:= O(k)$ many buckets. For each hash bucket $j \in [w]$, the algorithm maintains a continual release $\theta$-heavy hitter instance $\textsf{DP-HH}_j$.  We use the $\ell_1$-heavy hitter algorithm from~\cite{EMMMVZ23} as $\textsf{DP-HH}$ --- it returns a set $H$ of $\theta$-heavy hitters and their approximate counts $\hat{f}(\bfc)$ for all $\bfc \in H$. Since we are storing the centerpoints of all the cells marked as heavy hitters as candidate centers, we need to ensure that we do not store too many false positives, i.e., cells whose counts are much smaller than $\theta$ times the length of the hashed substream. To address this challenge, we have an additional pruning step that eliminates any cell $\bfc$ whose approximate count is less than  $\Theta(\theta) \hat{T}_{h(\bfc)}$ where $\hat{T}_{h(\bfc)}$ denotes the DP count of each hash bucket $j \in [w]$ at timestep $t \in [T]$. We keep track of $\hat{T}_{j}$ via an instance of the Binary Mechanism~\cite{DworkNPR10} denoted as $\textbf{BM}_j$ for each $j\in [w]$. Finally, only the centerpoints of cells that pass this pruning step are added as candidate centers to the set $\hat{\cF}_t$.

\subsection{Proof of \cref{thm:final-acc-dp-hh}} \label{sec:final-acc-dp-hh}

\begin{lemma}[Privacy of \cref{thm:final-acc-dp-hh}]
\textsf{BicriteriaDPCenters} is $\eps$-DP under the continual release setting. 
\end{lemma}
\begin{proof}

For a fixed timestep $t$, an input point $x_t$ is passed to an instance of \textsf{DPFindCenters} (\cref{alg:dp-find-center-online}) where it is assigned to a specific cell for a specific level of the quadtree, and cells at the same level are disjoint. The cell containing $x_t$ is then hashed to a bucket which maintains a \textsf{DP-HH} instance. 

There are $\log (\Lambda)$ levels per quadtree, thus point $x_t$ is a member of $\log(\Lambda)$ cells in total. As there are $2 \log^2(k)$ parallel processes --- considering $\log(k)$ quadtrees and $\log(k^2)$ parallel processes per quadtree --- a single point participates in $2\log \Lambda \log^2(k)$ total calls to \textsf{DP-HH}. %Note that we do not account for the $O(k)$ buckets that the cells are hashed into, as \textsf{DP-HH} is called on disjoint inputs for each bucket. 
Since \textsf{DP-HH} is DP under the continual release setting (by~\cref{thm:dp-hh-cr}), assigning each \textsf{DP-HH} instance a privacy budget of $\frac{\eps}{4\log(\Lambda) \log^2(k)}$ preserves $(\eps/2)$-DP by sequential composition (Item~\ref{sequential} of \cref{thm:composition}). 

Next, the algorithm uses the Binary Mechanism (which we know is DP by~\cref{thm:binarymechanism}) with a privacy budget of $\frac{\eps}{4\log(\Lambda) \log^2(k)}$ to keep track of the size of each hash substream $\hat{T}_j$ $\forall j \in [w]$. Since the input cells (and corresponding points within cells) are disjoint in each substream due to hashing, this preserves $\frac{\eps}{4\log(\Lambda) \log^2(k)}$-DP by parallel composition (Item~\ref{parallel} of \cref{thm:composition}) which over $2\log(\Lambda) \log^2(k)$ parallel processes preserves $(\eps/2)$-DP by sequential composition (Item~\ref{sequential} of \cref{thm:composition}). 

Finally, since only the centerpoints of cells identified as heavy hitters by \textsf{DP-HH} are released, these centers maintain $\eps$-DP by sequential composition (Item~\ref{sequential} of \cref{thm:composition}).

\end{proof}

\begin{lemma}[Accuracy of \cref{thm:final-acc-dp-hh}]
With probability at least $1-1/k^{O(\poly(k,\log(\Lambda))}$, 
    \begin{align*}
     cost(\cF, \cS) &\leq  O(d^3)cost(C^{opt}_{\cS}, \cS) + \tilde{O}\left(\frac{d^{2} \Lambda^2 k}{\eps} \poly\left(\log\left({T\cdot k \cdot\Lambda}\right)\right)\right)
\end{align*}
where $cost(C^{opt}_{\cS}, \cS)$ is the optimal $k$-means cost for $\cS$. % for all points $\cS:= \{x_1, \ldots, x_T\}$.
\end{lemma}
\begin{proof}
We first state some geometric properties regarding the cells within the quadtree construction. 
\begin{proposition}~\cite{EZNMC22}\label{prop:quad-ball}
Let $\cB$ be an $\ell_\infty$ ball of radius $r$ contained in $[-\Lambda, \Lambda]^d$ (it forms a $d$-dimensional cube with each side length $2r$). Then for a randomly shifted quadtree and any level $\ell$ with grid size at least $r' \geq 2r$, $\cB$ is split by the grid in each dimension $j \in [d]$ independently with probability $\frac{2r}{r'}$. 
\end{proposition}

Let $C^{opt}_\cS =\{c_1, \ldots, c_k\}$ be the optimal set of $k$ centers for the input set $\cS=\{x_1, \ldots, x_T\}$. For any radius, define $n_r$ as the number of points $x \in \cS$ such that $d(x,C^{opt}_\cS)\geq r$. 
Note that the opt cost of $k$-means (and $k$-median) is given by $\sum_{p \in \bbZ} 2^{2p} \cdot n_{2^p}$ and $ \sum_{p \in \bbZ} 2^{p} \cdot n_{2^p}$ (up to an $O(1)$-approximation).  

Fix some radius $r=2^p$ where $p \in \bbZ$ and consider a randomly shifted grid of size $20rd$. The following lemma characterizes cells containing $\cup^k_{i=1} \cB(c_i,r)$ with respect to the grid size. 

\begin{lemma}\label{lem:balls-cells}\cite{EZNMC22}
$\cup^k_{i=1} \cB(c_i,r)$ is contained in at most $4k$ cells of grid length $20rd$ by the corresponding level of the quadtree with probability at least 1/2.
\end{lemma}

Let $G_\ell$ where $0 \leq \ell \leq \log(\Lambda)$ be the set of $4k$ good cells of length $20rd$ (equivalently $\ell_2$-radius of $10rd^{3/2}$) at level $\ell$. Let the number of points in $\cS$ uncovered by $G_\ell$ be $n_{G_\ell}$. Observe that by \cref{lem:balls-cells}, since $G_\ell$ contains $\cup^k_{i=1} \cB(c_i,r)$ with probability at least 1/2, we have that $n_{G_\ell}\leq n_r$. It follows that 
\begin{align}\label{eq:len-radius-opt}
  &\sum^{\log(\Lambda)}_{\ell=0}(\text{grid length at level }\ell)^2 \cdot n_{G_\ell} \nonumber\\ 
  &\leq  O(d^{3})\sum_{p \in \bbZ\ :\ r=2^p \leq \Lambda} r^2 \cdot n_r  \leq O(d^3) \cdot  cost(C^{opt}_{\cS}, \cS)
\end{align}
Observe that we can define a one-one mapping between the level $\ell$ and the radius $r$, i.e., the radius $r$ (ranging from $1$ to $\Lambda$) maps to the grid length of a cell which is at most $\Lambda/2^\ell$ (level $\ell$ ranges from $\log (\Lambda)$ to 0). Since the grid length of a cell in $G_\ell$ at level $\ell$ is $20rd$ which maps to $20d \frac{\Lambda}{2^\ell}$, we can replace the leftmost term in the expression above as follows
\begin{align}\label{eq:len-radius-opt-1}
  O(d^2)\sum^{\log(\Lambda)}_{\ell=0}(\Lambda/2^\ell)^2  n_{G_\ell} &\leq   O(d^{3})\sum_{p \in \bbZ\ :\ r=2^p \leq \Lambda} r^2 \cdot n_r 
  \leq O(d^3) \cdot cost(C^{opt}_{\cS}, \cS)
\end{align}
Recall that we define $\cF_t$ as the set of centers until time step $t$. For a fixed level $\ell$, let the set of cells the algorithm \textsf{DP-HH} marks as heavy at timestep $t$ at level $\ell$ as $H_{\ell,t}$. Note that although there is an extra pruning step in \textsf{DPFindCenters} after the cells are marked heavy by \textsf{DP-HH}, we do not account for this here --- as if a cell is an $\alpha$-HH and marked heavy by \textsf{DP-HH}, and it survives the pruning step, it will still be an $\alpha$-HH. Then,
\begin{align*}
    cost(\cF_t) \leq O(d^{2})\sum^{\log(\Lambda)}_{\ell =0}(\Lambda/2^\ell)^2 \cdot \mathbbm{1}[x_t \text{ uncovered by }H_{\ell,t}] 
\end{align*}
Observe that 
\begin{align}\label{eq:cost-f}
    &cost(\cF) \nonumber\\
    &= \sum^T_{t=1} cost(\cF_t) \nonumber\\
    &\leq O(d^{2})\sum^{\log(\Lambda)}_{\ell =0}(\Lambda/2^\ell)^2 \cdot \sum^T_{t=1}\mathbbm{1}[x_t \text{ uncovered by }H_{\ell,t}] 
\end{align}

\begin{lemma}\label{lem:hh-uncover}
For a fixed level $\ell$, with probability at least $1-\frac{12}{k}$,
\begin{align}\label{eq:hli-upperbd}
    \sum^T_{t=1} \mathbbm{1}[x_t \text{ uncovered by }H_{\ell,t}] &\leq  (1+\theta)n_{G_\ell} 
    + \frac{4k \log^2 \Lambda \log k}{\eps \eta} \poly\left(\log\left(\frac{T\cdot 2^\ell}{\theta \xi \gamma_h }\right)\right)&
\end{align}
\end{lemma}
\begin{proof}
Observe that 
\begin{align*}
    &\sum^T_{t=1} \mathbbm{1}[x_t \text{ uncovered by }H_{\ell,t}] \nonumber\\
    &=\sum^T_{t=1} (\mathbbm{1}[(x_t \text{ uncovered by }H_{\ell,t})\land(x_t \text{ uncovered by }G_{\ell})]\nonumber \\&+\mathbbm{1}[(x_t \text{ uncovered by }H_{\ell,t})\land(x_t \text{ covered by }G_{\ell})]) \nonumber\\
    &=\sum^T_{t=1} \mathbbm{1}[(x_t \text{ uncovered by }H_{\ell,t})\land(x_t \text{ uncovered by }G_{\ell})] \nonumber \\&+\sum^T_{t=1} \mathbbm{1}[(x_t \text{ uncovered by }H_{\ell,t})\land(x_t \text{ covered by }G_{\ell})] %\label{eq:sum-uncoverh-coverg}
\end{align*}
The first sum in the above expression can be upper bounded by $n_{G_\ell}$, thus it remains to bound the second sum. In order to bound the second sum, we will need some properties of good cells that are hashed to buckets in \textsf{DPFindCenters}, the proofs of these claims can be found in \cref{app:bicriteria-proofs}. In the sequel, we denote $N_{\ell,\bfc}$ as the number of points in the cell $\bfc$ at level $\ell$ and the hash buckets as $\cB_j$ where $j \in [w]$. For simplicity, we consider the number of hash buckets $w:= 40k$. We first show that for any good cell $\bfc$, it is unlikely that the bucket it is hashed to contains another good cell $\bfc'\neq \bfc$. 
\begin{proposition}\label{prop:good-cell-collide}
Let $\bfc\in G_\ell$, then with probability at least $1/2$, for any $\bfc' \in G_\ell$ such that $\bfc'\neq \bfc$, we have that $ h(\bfc') \neq h(\bfc)$.
\end{proposition}
 In the next claim we give a bound on the size of the hash bucket denoted as $\cB_j$ where $j \in [w]$ in terms of the size of a good cell that is hashed to it and $n_r$. 
\begin{proposition}\label{prop:hash}
For each $\bfc \in G_\ell$, suppose the hash bucket $\cB_j$ where $j \in [w]$, contains only one good cell which is $\bfc$. Let $N_{\ell,\bfc}:= y$. Then with probability at least 1/2, $\vert \cB_j \vert \leq 2(y+\frac{n_{G_\ell}}{40k})$. 
\end{proposition}

Note that since the hashing procedure is run $\log(k^2)$ times in parallel, we can boost the success probabilities in the above claims to be $1-1/k^2$. 

Observe that for a fixed hash bucket $\cB_j$, any cell $\bfc$ such that $N_{\ell,\bfc} \geq \theta \cdot 2(y+\frac{n_{G_\ell}}{40k})$ qualifies as a $\theta$-heavy hitter since $N_{\ell,\bfc} \geq \theta \cdot 2(y+\frac{n_{G_\ell}}{40k}) \geq \theta \vert \cB_j\vert$ (by \cref{prop:hash}). In particular, for good cell $\bfc_y$ such that $N_{\ell,\bfc_y} = y$, if $\bfc_y$ is an $\theta$-HH then $y \geq \theta \cdot 2(y+\frac{n_{G_\ell}}{40k})$ or $ y \geq \frac{\theta n_{G_\ell}}{20k}$. We formalize this intuition in the claim below where we use the accuracy guarantees of \textsf{DP-HH} given by \cref{thm:dp-hh-cr} to characterize the good cells that are reported as $\theta$-HHs.

\begin{proposition}\label{prop:good-hh}
Let $\bfc \in G_\ell$. If $N_{\ell,\bfc} \geq \frac{\theta n_{G_\ell}}{20k}$, and $N_{\ell,\bfc} \geq  \frac{2\log(\Lambda) \log^2(k)}{\eps \gamma_h} \poly( \log(\frac{T\cdot k \cdot 2^\ell}{\theta \gamma_h}))$, then with probability at least $1-\frac{12}{k}$, $\bfc$ is reported as an $\theta$-heavy hitter by \textsf{DP-HH}. 
\end{proposition}

Finally, we give an upper bound for the number of points that are covered by good cells but for which \textsf{DP-HH} fails to report as heavy. 

\begin{proposition}\label{prop:uncovered-hh-covered-good}
With probability $1-\frac{12}{k}$,
\begin{align*}
    &\sum^T_{t=1} \mathbbm{1}[(x_t \text{ uncovered by }H_{\ell,t})\land(x_t \text{ covered by }G_{\ell})] \nonumber \\&\leq\theta n_{G_\ell} + \frac{8k \log(\Lambda) \log^2(k)}{\eps \gamma_h} \poly\left(\log\left(\frac{T\cdot k \cdot  2^\ell}{\theta \gamma_h }\right)\right)
\end{align*}
\end{proposition}

Thus by combining \cref{prop:uncovered-hh-covered-good} with our observation about the first sum being upper bounded by $n_{G_\ell}$ in the decomposition of $\sum^T_{t=1} \mathbbm{1}[x_t \text{ uncovered by }H_{\ell,t}]$, we obtain our desired statement in \cref{lem:hh-uncover}. 
\end{proof}

Note that we have shown \cref{lem:hh-uncover} is true with probability at least $1-\frac{12}{k}$, for a fixed level. Since we have $\log(\Lambda)$ many levels in a specific quadtree, and $\log(k)$ many quadtree instances in parallel --- we can boost our probability of success to be sufficiently high. It remains to bound the total $k$-means cost for the set of centers $\cF$ output by our algorithm. Combining \cref{eq:len-radius-opt}, \cref{eq:len-radius-opt-1} and \cref{eq:cost-f} along with \cref{lem:hh-uncover}, we obtain the following. 

\begin{align*}\small
    &cost(\cF, \cS) 
    \\&= \sum^T_{t=1} cost(\cF_t, \cS) \\
    &\leq O(d^{2})\sum^{\log(\Lambda)}_{\ell =0}(\Lambda/2^\ell)^2 \cdot \sum^T_{t=1}\mathbbm{1}[x_t \text{ uncovered by }H_{\ell,t}] \\
    &\leq O(d^{2})\sum^{\log(\Lambda)}_{\ell =0}(\Lambda/2^\ell)^2 \cdot ((1+\theta)n_{G_\ell}+  \frac{8k \log(\Lambda) \log^2(k)}{\eps \gamma_h} \poly\bigl(\log\bigl(\frac{T\cdot k \cdot 2^\ell}{\theta  \gamma_h }\bigr)\bigr)) \\
    &= O(d^{2})\sum^{\log(\Lambda)}_{\ell =0}(\Lambda/2^\ell)^2 \cdot (1+\theta)n_{G_\ell}+O(d^{2})\sum^{\log(\Lambda)}_{\ell =0}(\Lambda/2^\ell)^2 \cdot  \frac{8k \log(\Lambda) \log^2(k)}{\eps \gamma_h} \poly\bigl(\log\bigl(\frac{T\cdot k \cdot 2^\ell}{\theta \gamma_h }\bigr)\bigr) \\
    &\leq O(d^{3})(1+\theta) \sum_{p \in \bbZ\ :\ r=2^p \leq \Lambda} r^2 \cdot n_r + O(d^{2})\sum^{\log(\Lambda)}_{\ell =0}(\Lambda/2^\ell)^2 \cdot \frac{8k \log(\Lambda) \log^2(k)}{\eps \gamma_h} \poly\bigl(\log\bigl(\frac{T\cdot k \cdot 2^\ell}{\theta \gamma_h }\bigr)\bigr) \\
    &\leq O(d^3)(1+\theta) cost(C^{opt}_\cS,\cS) + O(d^{2})\sum^{\log(\Lambda)}_{\ell =0}(\Lambda/2^\ell)^2 \cdot  \frac{8k \log(\Lambda) \log^2(k)}{\eps \gamma_h}\poly\bigl( \log\bigl(\frac{T\cdot k \cdot 2^\ell}{\theta  \gamma_h }\bigr)\bigr)\\ 
  &=O(d^3)(1+\theta) cost(C^{opt}_\cS,\cS) + O\left(d^{2} \Lambda^2 \frac{k \log(\Lambda) \log^2(k)}{\eps \gamma_h} \poly\bigl(\log\bigl(\frac{T\cdot k \cdot  \Lambda}{\theta \gamma_h}\bigr)\bigr)\right)
\end{align*}

Finally, we can set $\theta$ (threshold for HHs) and $\gamma_h$ (approximation factor for frequency of a cell marked as heavy from \cref{thm:dp-hh-cr}) to appropriate constants. The accuracy claim follows. 
\end{proof}

\begin{lemma}[Space of \cref{thm:final-acc-dp-hh}]
\textsf{BicriteriaDPCenters} uses $O(k \log(\Lambda) \log^2(k) \poly\left( \log \left({T}{\Lambda}k\right)\right))$ space.
\end{lemma}
\begin{proof}
We analyze the total space usage for \textsf{DP-HH} in \cref{alg:dp-find-center-online} as this dominates space usage for the entire algorithm. From \cref{thm:dp-hh-cr}, one instance of \textsf{DP-HH} uses $\poly\left( \log \left({T}{\Lambda}k\right)\right)$. Since we run \textsf{DP-HH} on $O(k)$ many hash substreams, and $2\log(\Lambda) \log^2(k)$ parallel processes, the total space is $O(k \log(\Lambda) \log^2(k) \poly\left( \log \left({T}{\Lambda}k\right)\right))$.
\end{proof}

\begin{lemma}[Size of $\cF$ in \cref{thm:final-acc-dp-hh}]\label{lem:historical-hh}
For all $j\in [w]$, suppose the size of the hash bucket $\cB_j$ is $\vert \cB_j \vert = \Omega( \frac{\log(\Lambda) \log^3(k) }{\eps} \log^{2.5}T )$ then with high probability, the total number of heavy hitters at the end of the stream is $O({k}\log^2(k) \log(\Lambda) \log T)$. 

In other words, $\cF$ has atmost $O({k}\log^2(k) \log(\Lambda) \log T)$ centers.   
\end{lemma}

\begin{proof}
The algorithm runs independent instances of the \textsf{DP-HH} algorithm for each bucket of each level in each instantiation of the quadtree, thus it is sufficient to first show that for a fixed quadtree $Q$, a fixed level $\ell$, and a fixed bucket $\cB_j$ where $j\in [w]$, the total number of heavy hitters is at most $O(\frac{(1+\gamma_h)}{{\theta}}\log T)$.

Let the timestamps of points that end up in $\cB_j$ be $t_i=2^{i}$, where $0 \leq i \leq \log(T)$. Let the state of the hash bucket at time step $t$ be $\cB^{(t)}_j$. We set the failure probability in \cref{thm:binarymechanism} as $\xi:=\frac{1}{k^2}$. From \cref{thm:binarymechanism} we know that with probability $1-\frac{1}{k^2}$, the DP count of the hash bucket $\hat{T}_{j}$ at timestep $t$ has additive error $O(\frac{\log(\Lambda) \log^3(k)}{\eps} \log^{2.5}(T))$. Thus for a fixed timestamp $t$, if $\vert \cB^{(t)}_j \vert = \Omega(\frac{\log(\Lambda)\log^3(k)}{\eps} \log^{2.5}(T) )$, then we can see that a cell $\bfc$ is added to $\cF_t$ only if $\hat{f}(\bfc) > \frac{2\theta}{1000} \vert \cB^{(t)}_j \vert$. Recall from Condition 1 of \cref{thm:dp-hh-cr} that $\hat{N}_{\ell,\bfc} \in (1\pm \gamma_h) N_{\ell, \bfc}$. Thus if $\bfc \in \cF_t$ and $\vert \cB^{(t)}_j \vert = \Omega(\frac{\log(\Lambda)\log^3(k)}{\eps} \log^{2.5}(T) )$ then it must be the case that with probability $1-\frac{1}{k^2}$, $N_{\ell, \bfc} \geq \frac{{2\theta}}{1000(1+\gamma_h)}\| \cB^{(i)}_j \|_1 \geq \frac{{2\theta}}{1000(1+\gamma_h)} {t_{i-1}}$.     

Now, suppose for a contradiction, that the number of heavy hitters between $t_{i-1}$ and $t_i$ is at least $\frac{1000(1+\gamma_h)}{{\theta}}$. Then for each such cell $\bfc$, we have that $N_{\ell,\bfc}\geq \frac{{2\theta}}{1000(1+\gamma_h)} {t_{i-2}}$. Since there are at least $\frac{2000(1+\gamma_h)}{{2\theta}}$ such cells, this implies that the total number of points between $t_{i-1}$ and $t_i$ is $\geq \frac{2\theta}{1000(1+\gamma_h)} {t_{i-2}}\frac{2000(1+\gamma_h)}{{2\theta}} = 2^i = t_i$, which is a contradiction. Thus there must be at most $\frac{1000(1+\gamma_h)}{{\theta}}$ cells marked as heavy hitters between consecutive intervals, and since there are $\log(T)$ such intervals, we have that the total number of $\ell_1$ heavy hitters for a fixed bucket is $ O(\frac{(1+\gamma_h)}{{\theta}}\log T)$.

Boosting the success probability over $O(k)$ buckets, $\log(k^2)$ parallel processes, $\log(\Lambda)$ quadtree levels, and $\log(k)$ parallel processes of the quadtree instantiation, accounting for the additional number of heavy hitters, and taking $\theta$ and $\gamma_h$ as appropriate constants, we obtain the claim as stated.
\end{proof}

\section{DP Merge And Reduce Algorithm}\label{sec:dp-m-r}
 We give a differentially-private variant of the widely-known Merge and Reduce framework~\cite{har2004coresets,AgarwalHV04,feldman2020turning} that is used to efficiently release a coreset for a stream of points. The main idea behind the Merge and Reduce technique is to partition the input stream into blocks, compute a coreset for each block, take the union of the resulting coresets (merge step), and compute the coreset of the union (reduce step). The merging and reducing of the coresets is done in a tree-like fashion. In order to reduce the error introduced by merging, the number of levels in the tree must be small. On a high-level, our framework computes coresets at the base level (of the tree) using a DP semicoreset Algorithm $\cA$ (e.g.~\cite{StemmerK18,GhaziKM20}) and then computes coresets for subsequent levels using a non-DP coreset Algorithm $\cB$ (e.g.~\cite{Cohen-AddadSS21}).  
 
 First, we show that the semicoreset definition (i.e.,~\cref{def:semicoreset}) satisfies the Merge and Reduce properties, i.e., the union of semicoresets is a semicoreset and the coreset of a union of semicoresets is a valid semicoreset for the underlying points.

\begin{algorithm}[!htb]
\caption{Algorithm $\mathsf{DP\text{-}Merge\text{-}Reduce}$}
\label{alg:dp-merge-reduce}
%\begin{algorithmic}[1]
\KwData{ 
Input point $x'_t$ to be added to semicoreset, DP parameters $\eps,\eps_1,\delta_1$ 
}
%\Ensure{DP coreset for points in ring $R_{r}$ }
{\textbf{Initialize}($\eps,\eps_1,\delta_1$):}

%\If{$t=1$}
{Block size $M$\;}

\Comment{Value of $M$ depends on our choice of DP clustering algorithm. See \cref{thm:final-stemmer} and \cref{thm:final-ghazi} for exact values}

{$P_0=\emptyset$}\;

{Size of $P_0$ as ${p}_0=0$\;}%\Comment{Size of $P_0$}

\Comment{We only store $P_0$ in the actual algorithm, the rest of the $P_i$'s are only used in the accuracy analysis, so we do not run the following line in the actual algorithm. $N_r$ denotes the number of points in ring $R_r$.}

{$P_i=\emptyset$ for all $i=1, \ldots, u$ where $u= \lfloor \log(2{N}_{r}/M)\rfloor+1$\;} 

{$\gamma_i:= \gamma/Ci^2$ for $i=1, \ldots, u-1$\;}
\vspace{0.75em}   
{\textbf{Update}($x'_t$):}

\eIf{$x'_t \neq \perp$}{

{Insert $x'_t$ into $P_0$\;}

{$p_0 = p_0+1$}\Comment{update the size of $P_0$\;}

\eIf{$\textbf{LevelZero-AboveThreshold}(p_0, M, \eps) \to \top$}{
\Comment{See \cref{alg:dp-above-threshold}}

{$Q_1 \leftarrow \cA(P_0,k,d, \eps,\delta,\kappa , \eta_1,\eta_2)$} \Comment{$\cA$ computes a $(\eps_1,\delta_1)$-DP $(\kappa, \eta_1,\eta_2)$-semicoreset\;%, e.g.,~\cite{StemmerK18},\cite{GhaziKM20}
}
$P_1 = P_1 \cup P_0$\;

{Delete elements in $P_0$, and set ${p}_0=0$\;}

{$i=1$\;}

\While{$\exists\ Q'_i$ such that $Q_i$ and $Q'_i$ are both at level $i$}{
\Comment{\textbf{Merge\text{-}Reduce} computes $(1+\gamma_{i})$-coreset $Q_{i+1}$ of $Q_i\cup Q'_i$ and deletes $Q_i,Q'_i$ }

{$Q_{i+1}=\textbf{Merge\text{-}Reduce}(\cB, Q_i,Q'_i,\gamma_i)$\;} 

\Comment{$\cB$ computes a non-DP $(1+\gamma_i)$-coreset, e.g.,~\cite{Cohen-AddadSS21}}

{$P_{i+1}=P_{i+1}\cup P_i$\;}\label{li:p1}
\Comment{Current line and the one above are for analysis purpose only}

{Delete elements in $P_i$\;}\label{li:p2}

{$i=i+1$\;}
}
{return $\textbf{Merge\text{-}Reduce}(\cB, \{Q_i\}_{ i \leq u}, \gamma/3)$ \;}
}{
 %\Comment{when ${p}_0+\nu_i <\hat{M}_{count}$,}
{return $\emptyset$\;}
}
}{
{return $\emptyset$\;}
}
%\end{algorithmic}
\end{algorithm}

\begin{algorithm}[!htb]
\caption{Algorithm $\textsf{LevelZero-AboveThreshold}$}
\label{alg:dp-above-threshold}
%\begin{algorithmic}[1]
\KwData{Size of set $P_0$ denoted as $p_0$, Block size $M$, Privacy parameter $\eps$}
\KwResult{$\top$ if the size of set $P_0$ exceeds noisy threshold and $\perp$ otherwise}

{$\hat{M} = M+\Lap(2/\eps)$\;}
{$\nu = \Lap(4/\eps)$\;}
\eIf{$p_0+{\nu} \geq \hat{M}$}{
%\State{$\hat{M} = M+\Lap(2/\eps)$}
{return $\top$\;}
}{ 
{return $\bot$\;}
}
%\end{algorithmic}
\end{algorithm}

\begin{lemma}\label{lem:merge-reduce-property}
\begin{enumerate}
    \item (Merge) If $Q$ is a $(1+\gamma,\eta_1,\eta_2)$-semicoreset of $P$, $Q'$ is a $(1+\gamma,\eta_1,\eta_2)$-semicoreset of $P'$ and $P,P'$ are disjoint, then $Q \cup Q'$ is a $(1+\gamma, \eta_1,2\eta_2)$-semicoreset of $P \cup P'$. 
    \item (Reduce) If $R$ is a $(1+\gamma)$-coreset of $Q\cup Q'$, then $R$ is a $((1+\gamma)^2,(1-\gamma)\eta_1, (1+\gamma)2\eta_2)$-semicoreset of $P$.
\end{enumerate}
\end{lemma}
\begin{proof}
We first prove the merge property. For any set of $k$-centers $C \subseteq \bbR^d$, 
\begin{align*}
    \textsf{cost}(C,Q\cup Q') &\leq \textsf{cost}(C,Q) + \textsf{cost}(C,Q') \\
    &\leq (1+\gamma)\textsf{cost}(C,P) + \eta_2 + (1+\gamma)\textsf{cost}(C,P') + \eta_2 \\
    &\leq (1+\gamma)\textsf{cost}(C,P \cup P') + 2\eta_2 
\end{align*}

\begin{align*}
    \textsf{cost}(C,Q\cup Q') &\geq \textsf{cost}(C,Q) + \textsf{cost}(C,Q') \\
    &\geq (1-\gamma)\textsf{cost}(C,P) - \eta_1 \textsf{cost}(C^{opt}_P,P)-\eta_2 \nonumber \\&+ (1-\gamma)\textsf{cost}(C,P') - \eta_1 \textsf{cost}(C^{opt}_{P'},P')-\eta_2 \\
    &\geq (1-\gamma)\textsf{cost}(C,P \cup P') - \eta_1(\textsf{cost}(C^{opt}_P,P)+\textsf{cost}(C^{opt}_{P'},P'))- 2\eta_2\\
    &\geq (1-\gamma)\textsf{cost}(C,P \cup P') - \eta_1\textsf{cost}(C^{opt}_{P\cup P'},P\cup P')- 2\eta_2
\end{align*}

Next we prove the Reduce property below. For any set of $k$-centers $C \subseteq \bbR^d$,  
\begin{align*}
    \textsf{cost}(C,R) &\leq (1+\gamma)\textsf{cost}(C,Q \cup Q') \\
    &\leq (1+\gamma)((1+\gamma)\textsf{cost}(C,P \cup P') + 2\eta_2 )\\
    &\leq (1+\gamma)^2\textsf{cost}(C,P \cup P') + (1+\gamma)2\eta_2  \\
 %   &\leq (1+\gamma)^2\textsf{cost}(P \cup P') + (1+\gamma)2\eta\\
\end{align*}
\begin{align*}
    \textsf{cost}(C,R) &\geq (1-\gamma)\textsf{cost}(C,Q \cup Q') \\
    &\geq (1-\gamma)( (1-\gamma)\textsf{cost}(C,P \cup P') - \eta_1\textsf{cost}(C^{opt}_P,P\cup P')- 2\eta_2 )\\
    &\geq (1-\gamma)^2\textsf{cost}(C,P \cup P') - (1-\gamma)\eta_1\textsf{cost}(C^{opt}_P,P\cup P')-(1-\gamma)2\eta_2  \\
 %   &\leq (1+\gamma)^2\textsf{cost}(P \cup P') + (1+\gamma)2\eta\\
\end{align*}

\end{proof}

Observe that at any timestep the total set of input points seen so far (denoted as $P$) is partitioned into subsets $P_0,\ldots, P_u$ where some $P_i$'s can be empty and $u=\lfloor \log(2{N}/M)\rfloor+1$. Note that we simulate this step of partitioning $P$ into $P_i$'s in \cref{alg:dp-merge-reduce} solely for the analysis. It is not necessary to store $P_1, \ldots, P_u$ explicitly in the actual algorithm. 

We first prove some claims about \textsf{LevelZero-AboveThreshold}. 

\begin{lemma}\label{lem:dp-above-threshold}[Privacy] \textsf{LevelZero-AboveThreshold} is $\eps$-DP under the continual release setting.
\end{lemma}
\begin{proof}
Recall that \textsf{LevelZero-AboveThreshold} checks whether the size of the set $P_0$ is above a certain threshold. In other words, it checks whether the total count of the \emph{group} of elements $x_i \in P_0$ is above the given threshold. Once a positive response is returned by \textsf{LevelZero-AboveThreshold}, elements in $P_0$ are deleted and the process is repeated with a new group of elements. This algorithm is equivalent to grouping a stream of counts in~\cite{EMMMVZ23}. In particular, the proof of privacy is identical except that in our algorithm we do not release the actual noisy counts but instead we just release a positive/negative response to the same query.    
\end{proof}

\begin{lemma}[Accuracy of \textsf{LevelZero-AboveThreshold}]\label{lem:noise-above-threshold} For all $t\in [T]$, with probability $1-\xi$, we have that 
\begin{enumerate}
    \item $\vert \nu \vert < \frac{4}{\eps}\log(\frac{2T}{\xi})$
    \item $\vert \hat{M}-M \vert < \frac{2}{\eps} \log(\frac{2T}{\xi})$
\end{enumerate}
\end{lemma}
\begin{proof}
This follows from standard application of tail bounds for Laplace distribution and union bound over all $t \in [T]$. 
\end{proof}
\begin{lemma}\label{lem:p0-range}
If $p_0 + \nu \geq \hat{M}$ then with probability $1-\xi$, we have that $p_0 \geq M/2$.
\end{lemma}
\begin{proof}
We can simplify $p_0 + \nu \geq \hat{M}$ by applying the noise bounds from \cref{lem:noise-above-threshold}. Thus with probability at least $1-\xi$, we have that $p_0 \geq {M} - \frac{6}{\eps} \log(\frac{2T}{\xi})$. Finally using the assumption about $M>\frac{12}{\eps}\log(\frac{2T}{\xi})$, the statement follows. 
\end{proof}

\begin{lemma}
The number of levels is given by $u=\lceil \log(2N/M)\rceil+1$. 
\end{lemma}
\begin{proof}
By \cref{lem:p0-range}, $p_0 \geq M/2$, thus the total number of blocks at level 0 is $\leq \frac{2{N}}{M}$. The statement follows.  
\end{proof}

\subsection{Proof of \cref{thm:dp-merge-reduce}}\label{sec:dp-merge-reduce}

\begin{lemma}[Privacy of \cref{thm:dp-merge-reduce}] \label{lem:dpmr-priv}
\textsf{DP-Merge-Reduce} framework is $(\eps_1+\eps,\delta_1)$-DP under the continual release setting. 
\end{lemma}
\begin{proof}
First, by \cref{lem:dp-above-threshold}, \textsf{LevelZero-AboveThreshold} is $\eps$-DP under the continual release setting. Since the semicoresets computed by $\cA$ at level 1 are $(\eps_1,\delta_1)$-DP, we can release these DP semicoresets as they are computed. Subsequent computations on these DP semicoresets via non-DP algorithm $\cB$ preserve DP by postprocessing. Thus by sequential composition (Item~\ref{sequential} of \cref{thm:composition}) the claim follows.
\end{proof}

\begin{lemma}[Accuracy of \cref{thm:dp-merge-reduce}] \label{lem:dpmr-semicore-acc}
With probability at least $1-\xi-\xi_A-\xi_B$, \textsf{DP-Merge-Reduce} framework releases a $((1+\gamma)\kappa,(1-\gamma)\eta_1, (\frac{4N}{M}-1)(1+\gamma)\eta_2+\tilde{M})$-semicoreset of $P$. Where $\tilde{M}:= {M}+\frac{6}{\eps}\log(\frac{2T}{\xi})$.
\end{lemma}

\begin{proof} Recall that $P$ is partitioned by $P_1,\ldots, P_u$. We first prove the following claim about the coreset for a non-empty subset $P_r \subseteq P$.
\begin{lemma}
Suppose $P_r$ is non-empty. Then $Q_r$ is a $((1+\gamma/3)\kappa,(1-\gamma/3)\eta_1,(1+\gamma/3)2^{r-1}\eta_2 )$-semicoreset of $P_r$. 
\end{lemma}
\begin{proof}
We will first prove the claim that $Q_r$ is a $(\prod^{r-1}_{j=1}(1+\gamma_j)\kappa,\prod^{r-1}_{j=1}(1-\gamma_j)\eta_1, \prod^{r-1}_{j=1}(1+\gamma_j)2^{r-1}\eta_2 )$-semicoreset for $P_r$ where $r\geq 2$ by induction. Note that $Q_0=P_0$ for $p_0 + \nu < \hat{M}$. For $p_0 + \nu \geq \hat{M}$, $P_1 = P_1 \cup P_0$. Since we apply DP semicoreset algorithm $\cA$ to $P_1$, the resulting coreset $Q_1$ is a $(\kappa,\eta_1,\eta_2)$-semicoreset for ${P}_1$. 

{\bfseries Base Case. }
By \cref{lem:merge-reduce-property}, $Q_1 \cup Q'_1$ is a $(\kappa, \eta_1, 2\eta_2)$-semicoreset for $P_1 \cup P'_1$ and $Q_2$ is a $(\kappa(1+\gamma_1), (1+\gamma_1)\eta_1, (1+\gamma_1)2\eta_2)$-semicoreset of $Q_1\cup Q'_1$.  Here we use the notation $Q'_i$ or $P'_i$ to differentiate between sets or coresets at the same level $i$. 

{\bfseries Inductive Hypothesis. }Suppose the claim is true for $r=i$, i.e., $Q_i$ is a $(\kappa \prod^{i-1}_{j=1}(1+\gamma_j),\prod^{i-1}_{j=1}(1-\gamma_j)\eta_1,\prod^{i-1}_{j=1}(1+\gamma_j)2^{i-1}\eta_2 )$-semicoreset for $P_i$ and $Q'_i$ is a $(\kappa \prod^{i-1}_{j=1}(1+\gamma_j),\prod^{i-1}_{j=1}(1-\gamma_j)\eta_1, \prod^{i-1}_{j=1}(1+\gamma_j)2^{i-1}\eta_2 )$-semicoreset for $P'_i$. 

By \cref{lem:merge-reduce-property}, the Merge step implies
\begin{align}
    \mathsf{cost}(C,Q_i \cup Q'_i) &\leq \kappa\prod^{i-1}_{j=1}(1+\gamma_j)\mathsf{cost}(C,P_i \cup P'_i)+\prod^{i-1}_{j=1}(1+\gamma_j)2^{i} \eta_2 \\
    \mathsf{cost}(C,Q_i \cup Q'_i) &\geq \kappa\prod^{i-1}_{j=1}(1-\gamma_j)\mathsf{cost}(C,P_i \cup P'_i)\nonumber \\& - \prod^{i-1}_{j=1}(1-\gamma_j)\eta_1\mathsf{cost}(C^{opt}_{P_i \cup P'_i},P_i \cup P'_i) - \prod^{i-1}_{j=1}(1-\gamma_j)2^{i} \eta_2
\end{align}
The Reduce step implies that the resulting $(1+\gamma_{i+1})$-coreset of $Q_i \cup Q'_i$ denoted as $Q_{i+1}$ is such that
\begin{align}
    \mathsf{cost}(Q_{i+1}) &\leq \kappa\prod^{i}_{j=1}(1+\gamma_j)\mathsf{cost}(P_{i+1})+\prod^{i}_{j=1}(1+\gamma_j)2^{i}\eta_2\\
    \mathsf{cost}(Q_{i+1}) &\geq \kappa\prod^{i}_{j=1}(1-\gamma_j)\mathsf{cost}(P_{i+1})- \prod^{i-1}_{j=1}(1-\gamma_j)\eta_1\mathsf{cost}(C^{opt}_{P_{i+1}},P_{i+1}) -\prod^{i}_{j=1}(1-\gamma_j)2^{i}\eta_2
\end{align}

Finally, provided $c$ is large enough, we have: 
\begin{align}\label{eq:gamma-gamma'}
    &\prod^i_{j=1} (1+\gamma_j) \leq \prod^i_{j=1} \exp(\frac{\gamma}{c j^2}) = \exp(\frac{\gamma}{c}\sum^i_{j=1} \frac{1}{j^2}) \leq  \exp(\frac{\gamma}{c}\cdot \frac{\pi^2}{6}) \leq 1+\gamma/3  \\
    &\prod^i_{j=1} (1-\gamma_j) \leq \prod^i_{j=1} \exp(-\frac{\gamma}{c j^2}) = \exp(-\frac{\gamma}{c}\sum^i_{j=1} \frac{1}{j^2}) \leq  \exp(-\frac{\gamma}{c}\cdot \frac{\pi^2}{6}) \leq 1-\gamma/3  
\end{align}

The statement follows. 

\end{proof}
Finally, we release a $(1+\gamma/3)$-coreset of $\cup_{i \leq u}Q_i$ of $\cup_{i\leq u}P_i$ for all non-empty $P_i$ which by similar arguments as above is a $((1+\gamma)\kappa, (1-\gamma)\eta_1, (1+\gamma)\eta_2 (2^{u}-1))$-semicoreset of $P$. The statement follows by plugging in the value for the total number of levels $u$.

{Note that if ${p}_0+\nu < \hat{M}$ then we do not release anything. So we have to account for an additional additive error of ${M}+\frac{6}{\eps}\log(\frac{2T}{\xi})$ w.p. $1-\xi$ in this case.} If ${p}_0+\nu \geq \hat{M}$, then we proceed by computing a DP coreset using Algorithm $\cA$ with failure probability $\xi_A$. We also compute all the coresets past the first level using Algorithm $\cB$ and a failure probability of $\xi_B/2u$, where $u$ is the number of levels. Thus for a fixed run of \textsf{DP-Merge-Reduce}, by a union bound, the total failure probability for this part is at most $\xi_B$.  
\end{proof}

\begin{lemma}[Space of \cref{thm:dp-merge-reduce}] \label{lem:dpmr-space}
\textsf{DP-Merge-Reduce} framework uses \sloppy$S_\cA(M, k,d,\eps,\delta,\kappa,\eta_1,\eta_2, \xi_A)+ S_\cB(SZ_\cA(M,k,d,\eps,\delta,\kappa,\eta_1,\eta_2, \xi_A),k,d, \gamma, \xi_B)+\lceil \log(2N/M) \rceil \cdot S_\cB(SZ_\cB(M,k,d,\gamma, \xi_B),k,d, \gamma, \xi_B) +3M/2$ space.  
\end{lemma}
\begin{proof}
First we need an upper bound for the size of the blocks at level 0, i.e., $p_0$ which is given by the contrapositive statement of the claim below. 
\begin{lemma}
If $p_0 \geq \frac{3M}{2}$ then with probability $1-\xi$, ${p}_0+\nu \geq \hat{M}$. 
\end{lemma}
\begin{proof}
By applying the noise bounds from \cref{lem:noise-above-threshold} and using the assumption that $p_0 \geq \frac{3M}{2}$, we have that with probability $1-\xi$, 
\begin{align*}
    p_0+\nu \geq \frac{3M}{2} - \frac{4}{\eps}\log(\frac{2T}{\xi}) > \frac{3M}{2} - \frac{M}{3} = M+\frac{M}{6} > M+\frac{2}{\eps}\log(\frac{2T}{\xi}) > \hat{M}
\end{align*}
\end{proof}

Thus, we only need to store at most %$u \cdot SZ_{\cB}(N,k,d,\gamma) +SZ_{\cA}(M,k,d,\eps, \delta, \kappa,\eta)+ 3M/2$ for the coresets and 
$3M/2$ points for the block of input points at level 0, plus the additional space required to execute the coreset construction. Note that the semicoreset computation of level 0 blocks consumes space $S_\cA(M, k,d,\eps,\delta,\kappa,\eta_1,\eta_2)$ and since the largest coreset construction wrt non-DP algorithm $\cB$ is the union of at most $u$ coresets that is reduced to a single coreset, and the largest input to non-DP coreset algorithm is the resulting DP semicoreset size --- the additional storage size is at most $S_\cB(SZ_\cA(M,k,d,\eps,\delta,\kappa,\eta_1,\eta_2),k,d, \gamma)+(u-1) \cdot S_\cB(SZ_\cB(M,k,d,\gamma),k,d, \gamma)$. Note that the resulting coreset size for non-DP coreset algorithm $\cB$ is independent of the input set size. 

\end{proof}

\begin{lemma}[Semicoreset Size of \cref{thm:dp-merge-reduce}] \label{dpmr-size}
The resulting semicoreset has size at most $\tilde{O}(k\log (k) \cdot \gamma^{-4})$.
\end{lemma}
\begin{proof}
Since $Q_1 \leftarrow \cA(P_0,k,d, \eps,\delta,\kappa , \eta)$, size of $Q_1$ is $SZ_\cA(M,k,d,\eps,\delta,\kappa,\eta_1,\eta_2)$. Now $Q_2, \ldots, Q_u$ are obtained by running non-DP algorithm $\cB$. In order to simplify our notation, we invoke the state of the art non-DP coreset algorithm as $\cB$ (see \cref{thm:nondp-coreset}). Thus the semicoreset size of $Q_i$ where $2 \leq i \leq u$ is $\tilde{O}(k\log k \cdot \gamma^{-4})$. Finally we take the union of the semicoresets $\cup_{1 \leq i \leq u}Q_i$ which has size at most  $SZ_\cA(M,k,d,\eps,\delta,\kappa,\eta_1,\eta_2)+ (u-1)O(k\log k \cdot \gamma^{-4})$ and then apply non-DP algorithm $\cB$ to obtain a semicoreset of size at most $\tilde{O}(k\log k \cdot \gamma^{-4})$.
\end{proof}

\section{Missing Proofs from Section \ref{sec:dp-framework}}\label{app:dp-framework-proofs}

\subsection{Space and Size }\label{app:framework-space}

We present the space (see \cref{lem:main-space}) and coreset size (see \cref{lem:coreset-size-main}) of \cref{alg:extend-cluster} in terms of the space and coreset size of the underlying DP semicoreset algorithm $\cA$ and non-DP coreset algorithm $\cB$ used in \textsf{DP-Merge-Reduce}. We use the algorithm from~\cite{Cohen-AddadSS21} as our non-DP coreset algorithm $\cB$ whose guarantees are given by~\cref{thm:nondp-coreset}.

\begin{lemma}\label{lem:main-space}
\cref{alg:extend-cluster} consumes
\begin{align*}%\label{eq:space}
\sloppy
&\log(\Lambda) \cdot (S_\cA(M, k,d,\eps,\delta,\kappa,\eta_1,\eta_2,\xi_\cA)+ S_\cB(SZ_\cA(M,k,d,\eps,\delta,\kappa,\eta_1,\eta_2,\xi_\cA),k,d, \gamma, \xi_B)+3M/2) \nonumber\\&+\tilde{O}(k \log(T/M)\log(\Lambda)\log k \cdot \gamma^{-4}) +O({k}\log(k) \log^2(\Lambda) \log T)\\&+O(k \log(\Lambda) \log^2(k) \poly\left( \log \left({T}{\Lambda}k\right)\right))
\end{align*}
space.
\end{lemma}

\begin{proof}

The last term in the above claim is the total space used by \textsf{BicriteriaDPCenters} and the second last term is the space used to store the bicriteria solution $\cF$ (see \cref{thm:dp-hh-cr}). We focus on proving that the \textsf{DP-Merge-Reduce} instances consume the space specified by the first term of the above claim. 

We sometimes abuse notation and omit the input for the non-DP coreset size $SZ_\cB$ as by \cref{thm:nondp-coreset} we know that $SZ_\cB(\cdot) = \tilde{O}(k \log(k) \cdot \gamma^{-4})$. For a fixed epoch $T_i$, recall that the set of centers $\cF_t$ is fixed for $t \in T_i$. According to \cref{alg:extend-cluster}, at timestep $t \in T_i$, we compute $ \log(\Lambda)$ instances of \textsf{DP-Merge-Reduce} in parallel. By \cref{thm:dp-merge-reduce}, since the space required to compute semicoreset $\hat{\cY}^{(t)}_{r}$ using \textsf{DP-Merge-Reduce} is %$S_\cA(M, k,d,\eps, \delta, \kappa,\eta) +u\cdot S_\cB( k,d, \gamma)+ 3M/2$
$S_\cA(M, k,d,\eps,\delta,\kappa,\eta_1,\eta_2,\xi_\cA)+ S_\cB(SZ_\cA(M,k,d,\eps,\delta,\kappa,\eta_1,\eta_2,\xi_\cA),k,d, \gamma,\xi_\cB)+\lceil \log(2N/M) \rceil \cdot S_\cB(SZ_\cB(\cdot),k,d, \gamma, \xi_\cB) +3M/2$, the total space at timestep $t \in T_i$ is 
\begin{align}
& \sum^{\log(\Lambda)}_{r=1}  (S_\cA(M, k,d,\eps,\delta,\kappa,\eta_1,\eta_2,\xi_\cA)+ S_\cB(SZ_\cA(M,k,d,\eps,\delta,\kappa,\eta_1,\eta_2,\xi_\cA),k,d, \gamma,\xi_\cB)\nonumber\\&+\lceil \log(2 N^{(t)}_r/M) \rceil \cdot S_\cB(SZ_\cB(\cdot),k,d, \gamma,\xi_\cB) +3M/2)  \nonumber \\
&=\log(\Lambda) \cdot (S_\cA(M, k,d,\eps,\delta,\kappa,\eta_1,\eta_2,\xi_\cA)+ S_\cB(SZ_\cA(M,k,d,\eps,\delta,\kappa,\eta_1,\eta_2,\xi_\cA),k,d, \gamma,\xi_\cB)+3M/2) \nonumber \\&+ S_\cB(SZ_\cB(\cdot),k,d, \gamma,\xi_\cB) \sum^{\log(\Lambda)}_{r=1} \lceil \log(2 N^{(t)}_r/M) \rceil \label{eq:sum-r-other-1}
\end{align}
We focus on the last term in \cref{eq:sum-r-other-1}. In particular, since the space used by $\cB$ is linear in the coreset size, we have that $S_\cB(SZ_\cB(\cdot),k,d, \gamma)=\tilde{O}(k \log k \cdot \gamma^{-4})$. Next we simplify the adjoining sum 
\begin{align}
    &\sum^{\log(\Lambda)}_{r=1} \lceil \log(\frac{2N^{(t)}_r}{M}) \rceil\leq \sum^{\log(\Lambda)}_{r=1} (\log(\frac{2N^{(t)}_r}{M}) +1) = \log\left(\frac{2^{\log(\Lambda)}}{M^{\log(\Lambda)}}\cdot \prod^{\log(\Lambda)}_{r=1}N^{(t)}_r\right) + \log(\Lambda) \nonumber\\
    &< \log\left(\frac{\Lambda \cdot 2^{\log(\Lambda)}}{M^{\log(\Lambda)}}\cdot T^{\log(\Lambda)}\right) = \log(\Lambda)(2+\log(T)-\log(M)) 
\end{align}
where we use the fact that the number of points in any ring $N_r$ cannot be larger than $T$ in the second last step. 

\end{proof}
%%%%%%%%%%%%%%%%%%%%%%%%%%%%%%%%%%%%%%%%%%%%%%%%%%%%%%%%%%%%%%%%%%%%%
\begin{lemma}\label{lem:coreset-size-main}
\cref{alg:extend-cluster} releases a coreset of size at most $\tilde{O}(k \log(k)\cdot \gamma^{-4})$.
\end{lemma}
%%%%%%%%%%%%%%%%%%%%%%% PODS %%%%%%%%%%%%%%%%%%%%%%%%%%%%%%%%%%%%%%%%%%%%%%%%
\begin{proof}
Consider a fixed epoch $T_i$. By \cref{thm:dp-merge-reduce}, the coreset $\hat{\cY}^{(t)}_{r}$ has size $\tilde{O}(k \log(k)\cdot \gamma^{-4})$ for $t\in T_i$. Since we run a non-DP coreset algorithm at the end of each epoch (whose theoretical guarantees are given by \cref{thm:nondp-coreset}), the size of the coreset at the end of the epoch is also $\tilde{O}(k \log(k)\cdot \gamma^{-4})$. Recall that an epoch is created every time a new center is added to $\cF$, therefore the total size of the coreset is given by $\vert \cF \vert \cdot\tilde{O}(k \log(k)\cdot \gamma^{-4})$. But since we apply another non-DP coreset algorithm to $\hat{\cY}$ before releasing it, the coreset size is $\tilde{O}(k \log(k)\cdot \gamma^{-4})$. 
\end{proof}

%%%%%%%%%%%%%%%%%%%%%%%%%%%%%%%%%%%%%%%%%%%%%%%%%%%%%%%%%%%%%%%%%%%%%%%%%%%%%%%%%%%%%
\subsection{Accuracy} \label{app:framework-accuracy}
We analyze the accuracy of \cref{alg:extend-cluster} in the sequel. We first state the accuracy guarantee of the DP semicoreset $\hat{\cY}^{(t)}_{r}$ released by $\textsf{DP-Merge-Reduce}_r$ for each ring $R^{(t)}_r$ at timestep $t \in [T]$ in \cref{lem:yhat-dp-merge-reduce-output-semicoreset}. We defer the detailed proofs to \cref{app:dp-framework-proofs}.

\begin{lemma}\label{lem:yhat-dp-merge-reduce-output-semicoreset}
Let $\hat{\cY}^{(t)}_{r}$ be the output of $\textsf{DP-Merge-and-Reduce}_r$ (see \cref{alg:extend-cluster}) and $N^{(t)}_{r}$ be the number of (non-empty) points in $R^{(t)}_r$ at timestep $t \in [T]$. Then $\hat{\cY}^{(t)}_{r}$ is a DP $((1+\gamma)\kappa,(1+\gamma)\eta_1, (\frac{4N^{(t)}_r}{M}-1)(1+\gamma)\eta_2+\tilde{M})$-semicoreset of $R^{(t)}_r$ at timestep $t$. Where  {$\tilde{M}:= {M}+\frac{6}{\eps}\log(\frac{2T}{\xi})$}. In other words, for any set of $k$ centers $\cC$, with probability $1-3\xi$,
\begin{align}\label{eq:yhat-dp-merge-reduce-output-semicoreset}
    &\frac{(1-\gamma)}{\kappa}\cdot \mathsf{cost}(\cC,R^{(t)}_{r})-  (1-\gamma)\eta_1 \cdot \textsf{cost}(\cC^{opt}_{R^{(t)}_r},R^{(t)}_r)\cdot  (2^r)^2 - ((\frac{4N^{(t)}_{r}}{M}-1)(1+\gamma)\eta_2 {+\tilde{M}})\cdot (2^r)^2 \nonumber \\&\leq \mathsf{cost}(\cC,\hat{\cY}^{(t)}_{r})\nonumber\\ &\leq (1+\gamma)\kappa \cdot \mathsf{cost}(\cC,R^{(t)}_{r}) + ((\frac{4N^{(t)}_{r}}{M}-1)(1+\gamma)\eta_2 {+\tilde{M}}) \cdot (2^r)^2
\end{align}
\end{lemma}

\begin{proof}
The lemma follows from the accuracy guarantees of \textsf{DP-Merge-Reduce} (see \cref{thm:dp-merge-reduce}) and the fact that the additive cost incurred is proportional to the radius of the ring $R_r$ which is given by $2^r$.
\end{proof}

Next, we show that the union of the DP semicoresets over all rings $\hat{\cY}$ is a semicoreset for the stream $\cS$.

\begin{theorem}\label{thm:yhat-semicoreset}
Given dimension $d$, clustering parameter $k$, arbitrary parameter $C_M$, non-DP $(1+\gamma)$-coreset, $O(d^3)$-approximate bicriteria solution from \cref{thm:final-acc-dp-hh}, DP $(\kappa,\eta_1,\eta_2)$-semicoreset, and privacy parameter $\eps$. Let $\hat{\cY}$ be the output of \cref{alg:extend-cluster} for stream $\cS=\{x_1, \ldots, x_T\}$. Then for any set of $k$ centers $\cC$, we have that with probability $1-\frac{1}{T^2}-\frac{1}{k^{O(\poly(k,\log(\Lambda))}}$, the following holds
\begin{align}\label{eq:yhat-semicoreset}
&\frac{(1-\gamma)^2}{\kappa} \mathsf{cost}(\cC,\cS)- ((1-\gamma)^2\eta_1\Lambda^2 + 16 C_M(1-\gamma^2))  \cdot \mathsf{cost}(\cC^{opt}_\cS, \cS)\\& -C_M(1-\gamma^2) \cdot V'(d,k,\eps,T,\Lambda) ) \nonumber - (1-\gamma)V''(d,k,\eps,T,\Lambda) \nonumber \\&\leq \mathsf{cost}(\cC,\hat{\cY})\nonumber\\ &\leq (1+\gamma)^2\kappa \mathsf{cost}(\cC,\cS) +16 C_M(1+\gamma)^2\mathsf{cost}(\cC^{opt}_\cS, \cS) \nonumber\\&+ C_M(1+\gamma)^2 \cdot V'(d,k,\eps,T,\Lambda) )+(1+\gamma)V''(d,k,\eps,T,\Lambda)
\end{align}
\sloppy where $V'(d,k,\eps,T,\Lambda) = O\left(\frac{\Lambda^2 k^2}{d\eps} \log^2(\Lambda) \log^4(k) \log(T) \poly\log\left({T\cdot k \cdot\Lambda}\right)\right)$, $V''(d,k,\eps, T,\Lambda)= O({k}\Lambda^2\log^2(k) \log(\Lambda) \log T)\cdot (\frac{\alpha\eta_2}{C_M}+\tilde{O}(\frac{1}{\eps} \cdot \poly(\log(T\cdot k\cdot \Lambda))$, and $\alpha:=O(d^3)$.
\end{theorem}
%\paragraph{Proof of \cref{thm:yhat-semicoreset}}
\begin{proof}
Fix an epoch $T_i$. Let $\hat{\cY}^{(t)}$ be the output of \cref{alg:extend-cluster} at timestep $t \in T_i$. Since the centers in $\cF\vert_{T_i}$ and consequently the rings $R_r$ are fixed, we can compute the cost of $\hat{\cY}^{(t)}$ by summing \cref{eq:yhat-dp-merge-reduce-output-semicoreset} from \cref{lem:yhat-dp-merge-reduce-output-semicoreset} over all rings $R_r$ where $1 \leq r \leq \log(\Lambda)$ as follows:
\begin{align}\label{eq:bjr-yt-fixed-semicoreset}
    &\frac{(1-\gamma)}{\kappa}\sum_r \mathsf{cost}(\cC,R^{(t)}_{r})-  (1-\gamma)\eta_1 \sum_r \textsf{cost}(\cC^{opt}_{R^{(t)}_r},R^{(t)}_r)\cdot  (2^r)^2 \nonumber\\&- (1+\gamma)\eta_2\sum_r(\frac{4N^{(t)}_{r}}{M}-1)\cdot (2^r)^2 -\sum_r\tilde{M}\cdot (2^r)^2 \nonumber \\&\leq \mathsf{cost}(\cC,\hat{\cY}^{(t)})\nonumber\\ &\leq (1+\gamma)\kappa \sum_r \mathsf{cost}(\cC,R^{(t)}_{r}) + (1+\gamma)\eta_2\sum_r(\frac{4N^{(t)}_{r}}{M}-1) +\sum_r \tilde{M} \cdot (2^r)^2
\end{align}

Next, we take the union of $\hat{\cY}^{(t)}$ over all timesteps in epoch $T_i$. Note that for a fixed epoch and fixed ring $R_r$, the additive error $\tilde{M}$ is incurred at most once since this error stems from privately testing if the number of points at the base level is larger than the block size $M$ (see \cref{alg:dp-above-threshold}) in $\textsf{DP-Merge-Reduce}_r$. Thus summing \cref{eq:bjr-yt-fixed-semicoreset} over all timesteps in $T_i$ gives us 
\begin{align}\label{eq:bjr-yt-all-semicoreset}
    &\frac{(1-\gamma)}{\kappa} \mathsf{cost}(\cC,\cS\vert_{T_i})-  (1-\gamma)\eta_1\Lambda^2  \textsf{cost}(\cC^{opt}_{\cS\vert_{T_i}},\cS\vert_{T_i})\nonumber \\&- (1+\gamma)\eta_2\sum_{t \in T_i}\sum_r(\frac{4N^{(t)}_{r}}{M}-1)\cdot (2^r)^2 -\tilde{M}\cdot\frac{4}{3}(\Lambda^2-1) \nonumber \\&\leq \mathsf{cost}(\cC,\hat{\cY}^{(t)}\vert_{t\in T_i})\nonumber\\ &\leq (1+\gamma)\kappa \mathsf{cost}(\cC,\cS\vert_{T_i}) + (1+\gamma)\eta_2\sum_{t \in T_i}\sum_r(\frac{4N^{(t)}_{r}}{M}-1) +\tilde{M} \cdot \frac{4}{3}(\Lambda^2-1)
\end{align}

where $\cS\vert_{T_i}$ denotes the input points in stream $\cS$ restricted to epoch $T_i$ and $\hat{\cY}^{(t)}\vert_{t\in T_i}$ is defined analogously.

\begin{lemma}\label{lem:radius-opt-cost}
If $x_i$ is assigned to the ring $R_{r}$, let $r(x_i):= 2^r$. Then 
\begin{align}
    \sum_{x_i \in \cS} r(x_i)^2 \leq \sum_{x_i \in \cS} 4d(x_i,\cF)^2 = 4\textsf{cost}(\cF,\cS) \leq 4\alpha \cdot \textsf{cost}(\cC^{opt}_\cS,\cS)+ V(d,k,\eps,T,\Lambda)
\end{align}
\sloppy where $\alpha:=O(d^3)$ and $V(d,k,\eps,T,\Lambda)=  O\left(\frac{d^{2} \Lambda^2 k \log(\Lambda) \log^2(k)}{\eps} \poly\bigl(\log\bigl({T\cdot k \cdot  \Lambda}\bigr)\bigr)\right)$ are the multiplicative/additive factors from the bicriteria approximation (see \cref{thm:final-acc-dp-hh}). 
\end{lemma}
\begin{proof}
The statement follows immediately from the definition of the ring $R_r$ (\cref{def:ring-center}) and the accuracy guarantees of the bicriteria solution (\cref{thm:final-acc-dp-hh}). 
\end{proof}

Observe that {$\sum_{t\in T_i}\sum_{r} N^{(t)}_{r}\cdot (2^r)^2 = \sum_{x \in \cS\vert_{T_i}} r(x)^2$}, thus we can use \cref{lem:radius-opt-cost} to simplify the additive error in \cref{eq:bjr-yt-all-semicoreset} to obtain 
\begin{align}\label{eq:bjr-yt-epoch-semicoreset}
    &\frac{(1-\gamma)}{\kappa} \mathsf{cost}(\cC,\cS\vert_{T_i})-  (1-\gamma)\eta_1\Lambda^2  \textsf{cost}(\cC^{opt}_{\cS\vert_{T_i}},\cS\vert_{T_i})\nonumber\\&- \frac{4}{M}(1+\gamma)\eta_2 (4\alpha \cdot \mathsf{cost}(\cC^{opt}_{\cS \vert_{T_i}}, \cS\vert_{T_i}) +V(d,k,\eps,T,\Lambda) )-\tilde{M}\cdot\frac{4 \Lambda^2}{3} \nonumber \\&\leq \mathsf{cost}(\cC,\hat{\cY}^{(t)}\vert_{t\in T_i})\nonumber\\ &\leq (1+\gamma)\kappa \mathsf{cost}(\cC,\cS\vert_{T_i})  + \frac{4}{M}(1+\gamma)\eta_2 (4\alpha \cdot \mathsf{cost}(\cC^{opt}_{\cS \vert_{T_i}}, \cS\vert_{T_i}) +V(d,k,\eps,T,\Lambda) ) +\tilde{M} \cdot \frac{4\Lambda^2}{3}
\end{align}

Since a new epoch starts whenever a new center is added to the bicriteria solution $\cF$, and we run a $(1+\gamma)$-approx non-DP coreset algorithm on the coreset $\hat{\cY}$ at the beginning of a new epoch, denote this new coreset as $\hat{\cY}_{T_i}$. We state the guarantees of $\hat{\cY}_{T_i}$ below.  
\begin{align}\label{eq:bjr-new-epoch-semicoreset}
    &\frac{(1-\gamma)^2}{\kappa} \mathsf{cost}(\cC,\cS\vert_{T_i})-  (1-\gamma)^2\eta_1\Lambda^2  \textsf{cost}(\cC^{opt}_{\cS\vert_{T_i}},\cS\vert_{T_i}) \nonumber\\&- \frac{4}{M}(1-\gamma^2)\eta_2 (4\alpha \cdot \mathsf{cost}(\cC^{opt}_{\cS \vert_{T_i}}, \cS\vert_{T_i}) +V(d,k,\eps,T,\Lambda) )-\frac{4 \Lambda^2 (1-\gamma)\tilde{M}}{3} \nonumber \\&\leq \mathsf{cost}(\cC,\hat{\cY}_{T_i})\nonumber\\ &\leq (1+\gamma)^2\kappa \mathsf{cost}(\cC,\cS\vert_{T_i}) + \frac{4}{M}(1+\gamma)^2\eta_2 (4\alpha \cdot \mathsf{cost}(\cC^{opt}_{\cS \vert_{T_i}}, \cS\vert_{T_i}) \nonumber \\&+V(d,k,\eps,T,\Lambda) )  +\frac{4\Lambda^2 (1+\gamma)\tilde{M}}{3}
\end{align}

Observe that the total number of epochs is bounded by the total number of centers in $\cF$ at the end of the stream. By \cref{thm:final-acc-dp-hh} we know that $\vert \cF \vert = O({k}\log^2(k) \log(\Lambda) \log T)$. Therefore, we can sum \cref{eq:bjr-new-epoch-semicoreset} over all epochs to have that for any set of $k$ centers $\cC$:
\begin{align}\label{eq:bjr-yt-epoch-all-semicoreset}
    &\frac{(1-\gamma)^2}{\kappa} \mathsf{cost}(\cC,\cS)-  (1-\gamma)^2\eta_1\Lambda^2  \textsf{cost}(\cC^{opt}_{\cS},\cS)- \frac{4}{M}(1-\gamma^2)\eta_2 (4\alpha \cdot \mathsf{cost}(\cC^{opt}_\cS, \cS) \nonumber\\&+V^*(d,k,\eps,T,\Lambda) )-{\vert \cF \vert} \cdot \frac{4 \Lambda^2 (1-\gamma)\tilde{M}}{3} \nonumber \\&\leq \mathsf{cost}(\cC,\hat{\cY})\nonumber\\ &\leq (1+\gamma)^2\kappa \mathsf{cost}(\cC,\cS) + \frac{4}{M}(1+\gamma)^2\eta_2 (4\alpha \cdot \mathsf{cost}(\cC^{opt}_\cS, \cS) \nonumber \\&+V^*(d,k,\eps,T,\Lambda) )+{\vert \cF \vert} \cdot \frac{4\Lambda^2 (1+\gamma)\tilde{M}}{3}
\end{align}

where $V^*(d,k,\eps,T,\Lambda) = O\left(\frac{d^{2} \Lambda^2 k^2}{\eps} \log^2(\Lambda) \log^4(k) \log(T) \poly\left(\log\left({T\cdot k \cdot\Lambda}\right)\right)\right)$. 

Grouping like terms together, we have 
\begin{align}\label{eq:bjr-yt-epoch-all-semicoreset-final}
     &\frac{(1-\gamma)^2}{\kappa} \mathsf{cost}(\cC,\cS)-  ((1-\gamma)^2\eta_1\Lambda^2 + \frac{16\alpha}{M}(1-\gamma^2)\eta_2)  \cdot \mathsf{cost}(\cC^{opt}_\cS, \cS) \nonumber\\&-\frac{4}{M}(1-\gamma^2)\eta_2 \cdot V^*(d,k,\eps,T,\Lambda) )-{\vert \cF \vert} \cdot \frac{4 \Lambda^2 (1-\gamma)\tilde{M}}{3} \nonumber \\&\leq \mathsf{cost}(\cC,\hat{\cY})\nonumber\\ &\leq (1+\gamma)^2\kappa \mathsf{cost}(\cC,\cS) + \frac{4}{M}(1+\gamma)^2\eta_2 (4\alpha \cdot \mathsf{cost}(\cC^{opt}_\cS, \cS) \nonumber \\&+V^*(d,k,\eps,T,\Lambda) )+{\vert \cF \vert} \cdot \frac{4\Lambda^2 (1+\gamma)\tilde{M}}{3}
\end{align}
 We set 
\begin{align}\label{eq:M}
 M:= \frac{\alpha \eta_2 }{C_M} 
\end{align} 
where $C_M$ is a parameter chosen in the sequel (see \cref{rem:param-M}). Simplifying \cref{eq:bjr-yt-epoch-all-semicoreset-final} and taking a union bound over all rings and epochs, the desired claim in the theorem statement holds with probability $1- \log(\Lambda)\vert \cF \vert 3 \xi-\frac{1}{k^{O(\poly(k,\log(\Lambda))}}$. Thus we set $\xi:=  \frac{1}{3\vert \cF \vert \log(\Lambda)T^2}$. 
\end{proof}

Finally, \cref{thm:final-accuracy-main} gives the cost of clustering result for the output $\hat{\cY}$ after the offline postprocessing step is executed. Recall that the postprocessing step involves running a $\rho$-approximation non-DP clustering algorithm on $\hat{\cY}$. We present the clustering guarantee in terms of the $\rho$-approx non-DP clustering algorithm, the  non-DP $(1+\gamma)$-coreset algorithm, the DP $(\kappa, \eta_1,\eta_2)$-coreset algorithm and parameter $C_M$. 

\begin{theorem}[Main Accuracy]\label{thm:final-accuracy-main}
Given dimension $d$, clustering parameter $k$, arbitrary parameter $C_M$, non-DP $(1+\gamma)$-coreset, $O(d^3)$-approximate bicriteria solution from \cref{thm:final-acc-dp-hh}, DP $(\kappa,\eta_1,\eta_2)$-semicoreset, and privacy parameter $\eps$. Let $\cC_{\hat{\cY}}$ be the set of $k$ centers obtained from running the offline $\rho$-approx non-DP $k$-means algorithm on $\hat{\cY}$. Then, 
 \begin{align}
  &\mathsf{cost}(\cC_{\hat{\cY}},\cS) \nonumber\\
  &\leq \frac{\kappa}{(1-\gamma)^3}\cdot ((1+\gamma)^3\rho (\kappa + 16 C_M) +((1-\gamma)^3\eta_1\Lambda^2 + 16C_M(1-\gamma^2)(1-\gamma)))\mathsf{cost}(\cC^{opt}_{\cS},\cS) \nonumber\\&+ \frac{\kappa}{(1-\gamma)^3}\cdot(\rho C_M(1+\gamma)^3+ C_M(1-\gamma^2)(1-\gamma)) V'(d,k,\eps,T,\Lambda)\nonumber \\&+ \frac{\kappa}{(1-\gamma)^3}\cdot((1+\gamma)^2\rho+ (1-\gamma)^2) V''(d,k,\eps,T,\Lambda)
\end{align}
\sloppy where $V'(d,k,\eps,T,\Lambda) = O\left(\frac{\Lambda^2 k^2}{d\eps} \log^2(\Lambda) \log^4(k) \log(T) \poly\log\left({T\cdot k \cdot\Lambda}\right)\right)$, $V''(d,k,\eps, T,\Lambda)= O({k}\Lambda^2\log^2(k) \log(\Lambda) \log T)\cdot \left(\frac{\alpha\eta_2}{C_M}+\tilde{O}(\frac{1}{\eps} \cdot \poly(\log(T\cdot k\cdot \Lambda)\right)$, and $\alpha:=O(d^3)$. \end{theorem}

%\paragraph{Proof of \cref{thm:final-accuracy-main}.}
 \begin{proof}
We first recall that our algorithm releases the semicoreset $\hat{\cY}$ after running a non-DP $(1+\gamma)$-coreset algorithm on it. Thus, the guarantee of the resulting semicoreset which we call $\hat{\cY}_{off}$ is 
\begin{align}\label{eq:yhat-semicoreset-opt}
&\frac{(1-\gamma)^3}{\kappa} \mathsf{cost}(\cC,\cS)-  ((1-\gamma)^3\eta_1\Lambda^2 + 16C_M(1-\gamma^2)(1-\gamma))  \cdot \mathsf{cost}(\cC^{opt}_\cS, \cS) \nonumber \\&- C_M(1-\gamma^2)(1-\gamma) \cdot V'(d,k,\eps,T,\Lambda) ) - (1-\gamma)^2V''(d,k,\eps,T,\Lambda) \nonumber \\&\leq \mathsf{cost}(\cC,\hat{\cY}_{off})\nonumber\\ &\leq (1+\gamma)^3\kappa \mathsf{cost}(\cC,\cS) + 16 C_M(1+\gamma)^3\mathsf{cost}(\cC^{opt}_\cS, \cS) \nonumber \\&+ C_M(1+\gamma)^3 \cdot V'(d,k,\eps,T,\Lambda) )+(1+\gamma)^2V''(d,k,\eps,T,\Lambda)
\end{align}
\sloppy where $V'(d,k,\eps,T,\Lambda) = O\left(\frac{\Lambda^2 k^2}{d\eps} \log^2(\Lambda) \log^4(k) \log(T) \poly\log\left({T\cdot k \cdot\Lambda}\right)\right)$, $V''(d,k,\eps, T,\Lambda)= O({k}\Lambda^2\log^2(k) \log(\Lambda) \log T)\cdot (\frac{\alpha\eta_2}{C_M}+\tilde{O}(\frac{1}{\eps} \cdot \poly(\log(T\cdot k\cdot \Lambda))$, and $\alpha:=O(d^3)$.

Recall from the theorem statement that $\cC_{\hat{\cY}}$ is the set of $k$ centers obtained from running the offline $\rho$-approx non-DP $k$-means algorithm on $\hat{\cY}_{off}$ which gives us the following guarantee,
 \begin{align}\label{eq:y-opt-sc}
       \mathsf{cost}(\cC_{\hat{\cY}},\hat{\cY}_{off}) \leq \rho \cdot \mathsf{cost}(\cC^{opt}_{\hat{\cY}}, {\hat{\cY}_{off}})
 \end{align}
Since \cref{eq:yhat-semicoreset-opt} is true for any set of $k$ centers, in particular it holds for $\cC_{\hat{\cY}}$, which gives us the following

\begin{align}\label{eq:opt-y-semicoreset}
 &\frac{(1-\gamma)^3}{\kappa} \mathsf{cost}(\cC_{\hat{\cY}},\cS)-  ((1-\gamma)^3\eta_1\Lambda^2 + 16C_M(1-\gamma^2)(1-\gamma))  \cdot \mathsf{cost}(\cC^{opt}_\cS, \cS) \nonumber \\&-C_M(1-\gamma^2)(1-\gamma) \cdot V'(d,k,\eps,T,\Lambda) ) - (1-\gamma)^2V''(d,k,\eps,T,\Lambda) \nonumber \\&\leq \mathsf{cost}(\cC_{\hat{\cY}},\hat{\cY}_{off})\nonumber\\ &\leq (1+\gamma)^3\kappa \mathsf{cost}(\cC_{\hat{\cY}},\cS) + 16 C_M(1+\gamma)^3\mathsf{cost}(\cC^{opt}_\cS, \cS) \nonumber \\&+ C_M(1+\gamma)^3 \cdot V'(d,k,\eps,T,\Lambda) )+(1+\gamma)^2V''(d,k,\eps,T,\Lambda)
\end{align}

We derive the desired expression in the theorem statement by starting with the LHS of \cref{eq:opt-y-semicoreset} and then applying \cref{eq:y-opt-sc} to obtain \cref{eq:rho-cy-sc}. By observing that the cost of $\hat{\cY}_{off}$ wrt its optimal set of centers $\cC^{opt}_{\hat{\cY}}$ is less than its cost wrt to $\cC^{opt}_{\cS}$, we obtain \cref{eq:opt-sc}. \cref{eq:coreset-app} is obtained by simply applying the RHS of the coreset relation from \cref{eq:opt-y-semicoreset} wrt $\cC^{opt}_{\cS}$. The final statement is derived by grouping like terms and simplifying \cref{eq:simplify-sc}. 
\begin{align}
    &\frac{(1-\gamma)^3}{\kappa} \mathsf{cost}(\cC_{\hat{\cY}},\cS)\\
    &\leq \mathsf{cost}(\cC_{\hat{\cY}},\hat{\cY}_{off}) +  ((1-\gamma)^3\eta_1\Lambda^2 + 16C_M(1-\gamma^2)(1-\gamma))  \cdot \mathsf{cost}(\cC^{opt}_\cS, \cS) \nonumber \\&+ C_M(1-\gamma^2)(1-\gamma) \cdot V'(d,k,\eps,T,\Lambda) ) + (1-\gamma)^2V''(d,k,\eps,T,\Lambda) \\
    &\leq \rho \mathsf{cost}(\cC^{opt}_{\hat{\cY}},\hat{\cY}_{off}) +  ((1-\gamma)^3\eta_1\Lambda^2 + 16C_M(1-\gamma^2)(1-\gamma))  \cdot \mathsf{cost}(\cC^{opt}_\cS, \cS) \nonumber \\&+ C_M(1-\gamma^2)(1-\gamma) \cdot V'(d,k,\eps,T,\Lambda) ) + (1-\gamma)^2V''(d,k,\eps,T,\Lambda) \label{eq:rho-cy-sc} \\
    &\leq \rho \mathsf{cost}(\cC^{opt}_{\cS},\hat{\cY}_{off}) +  ((1-\gamma)^3\eta_1\Lambda^2 + 16C_M(1-\gamma^2)(1-\gamma))  \cdot \mathsf{cost}(\cC^{opt}_\cS, \cS) \nonumber \\&+ C_M(1-\gamma^2)(1-\gamma) \cdot V'(d,k,\eps,T,\Lambda) ) + (1-\gamma)^2V''(d,k,\eps,T,\Lambda) \label{eq:opt-sc}\\
    &\leq \rho ((1+\gamma)^3\kappa \mathsf{cost}(\cC^{opt}_{\cS},\cS) + 16 C_M(1+\gamma)^3\mathsf{cost}(\cC^{opt}_\cS, \cS) + C_M(1+\gamma)^3 \cdot V'(d,k,\eps,T,\Lambda) \nonumber\\&+(1+\gamma)^2V''(d,k,\eps,T,\Lambda)) +  ((1-\gamma)^3\eta_1\Lambda^2 + 16C_M(1-\gamma^2)(1-\gamma))  \cdot \mathsf{cost}(\cC^{opt}_\cS, \cS) \nonumber \\&+ C_M(1-\gamma^2)(1-\gamma) \cdot V'(d,k,\eps,T,\Lambda) ) + (1-\gamma)^2V''(d,k,\eps,T,\Lambda) \label{eq:coreset-app}\\
    &\leq ((1+\gamma)^3\rho (\kappa + 16 C_M) +((1-\gamma)^3\eta_1\Lambda^2 + 16C_M(1-\gamma^2)(1-\gamma)))\mathsf{cost}(\cC^{opt}_{\cS},\cS) \nonumber\\&+(\rho C_M(1+\gamma)^3+ C_M(1-\gamma^2)(1-\gamma)) V'(d,k,\eps,T,\Lambda)+((1+\gamma)^2\rho + (1-\gamma)^2) V''(d,k,\eps,T,\Lambda) \label{eq:simplify-sc}
\end{align}

\end{proof}

\section{Missing Proofs from Appendix \ref{sec:bicriteria}}\label{app:bicriteria-proofs}
%\tnote{FIX MISSING REFS}

\paragraph{Proof of \cref{prop:good-cell-collide}. }
\begin{proof}
For any good cell $\bfc' \neq \bfc$, define $$X_{\bfc'}= \begin{cases} 1 & \text{if }\bfc'\text{ collides with }\bfc\\ 0 &\text{otherwise}\end{cases}$$
Now for a fixed $\bfc' \neq \bfc$, the expected number of collisions with $\bfc$ is given by $\E_h[X_{\bfc'}] = \Pr[X_{\bfc'}=1] =\Pr[h(\bfc)=h(\bfc')] = \frac{1}{w} = \frac{1}{40k}$. Thus the total number of collisions for $\bfc'$ with $\bfc$, in expectation, is given by $\E_h[\sum_{\bfc'} X_{\bfc'}] \leq 4k/w = 1/10 \leq 1/2$. 

Therefore, by Markov, $\Pr[\sum_{\bfc'} X_{\bfc'}>1] \leq 1/2$. The claim follows. 
\end{proof}
\paragraph{Proof of \cref{prop:hash}.}
\begin{proof}
\begin{align}
    \E_h \vert \cB_j \vert &= \E_h [ y + \sum_{\bfc \not\in G_\ell} N_{\ell, \bfc}]\\
    &=y + \E_h[\sum_{\bfc \not\in G_\ell} N_{\ell, \bfc}]\\
    &=y+ \frac{n_{G_\ell}}{40k}
\end{align}
The claim follows by Markov inequality. 
\end{proof}

\paragraph{Proof of \cref{prop:good-hh}. }
\begin{proof}
We set the failure probability of \textsf{DP-HH} (see \cref{thm:dp-hh-cr}) to be $\xi:=\frac{1}{k^2}$. 
For $\bfc \in G_\ell$, define $\cE_\bfc := \cE_1 \wedge \cE_2 \wedge \cE_3$ where $\cE_1$ is the event that \textsf{DP-HH} algorithm is correct on all instances, $\cE_2$ is the event that there are no collisions between $\bfc$ and other good cells, and $\cE_3$ is the event that there exists a hash bucket that contains only $\bfc$, and if $N_{\ell,\bfc}:= y$, then the size of that hash bucket is $\leq 2(y+\frac{n_{G_\ell}}{40k})$. We know that $\Pr[\cE_1] \geq 1-\xi= 1-\frac{1}{k^2}$, 
by the accuracy guarantee of \textsf{DP-HH} algorithm, $\Pr[\cE_2] \geq 1-\frac{1}{k^2}$ by \cref{prop:good-cell-collide} (and boosting the success probability), and $\Pr[\cE_3] \geq 1-\frac{1}{k^2}$, by \cref{prop:hash} (and boosting the success probability). Thus by a union bound, we have that for a fixed $\bfc \in G_\ell$, $\Pr[\cE_\bfc] \geq 1-\frac{3}{k^2}$. 
By taking a union bound over all $4k$ good cells, with probability at least $1-\frac{12}{k}$, the claim holds. 
\end{proof}

\paragraph{Proof of \cref{prop:uncovered-hh-covered-good}. }
\begin{proof}
Recall that $N^{(t)}_{\ell,\bfc}$ is the number of points in cell $\bfc$ at time step $t$. Because $x_t$ is covered in $G_\ell$, this means that we only need to care about the points covered by a good cell $\bfc \in G_\ell$ but $\bfc \not\in H_{\ell,t}$. Using \cref{prop:good-hh}, we know that if $N^{(t)}_{\ell,\bfc} \geq \frac{\theta n_{G_\ell}}{20k} $, and $N^{(t)}_{\ell,\bfc} \geq  \frac{2\log(\Lambda) \log^2(k)}{\eps \eta} \poly(\log(\frac{T\cdot k \cdot 2^\ell}{\theta \eta }))$, then with probability at least $1-\frac{12}{k}$, we have that $\bfc \in H_{\ell,t}$. 

If $\bfc \in G_\ell$, and $\bfc \not \in H_{\ell,t}$, this means that either (1) $N^{(t)}_{\ell,\bfc} <  \frac{2\log(\Lambda) \log^2(k)}{\eps \eta} \poly(\log(\frac{T\cdot k \cdot 2^\ell}{\theta \eta }))$, or, (2) $\bfc$ is not an $\theta$-HH. For case (1), since there are $4k$ such good cells, the total number of uncovered points in such cells are $\frac{4k \cdot 2\log(\Lambda) \log^2(k)}{\eps \eta} \poly(\log(\frac{T\cdot k \cdot 2^\ell}{\theta  \eta }))$. 
For case (2), this means that $N^{(t)}_{\ell,\bfc} < \frac{\theta n_{G_\ell}}{20k} $. 
Again, since there are $4k$ such good cells, the total number of uncovered points in such cells are $< \frac{\theta n_{G_\ell}}{20k} \cdot 4k < \theta n_{G_\ell} $. 
\end{proof}

\end{document}